\documentclass[10pt, conference, letterpaper]{IEEEtran}
\IEEEoverridecommandlockouts
\usepackage{cite}
\usepackage{amsmath,amssymb,amsfonts}
\usepackage{verbatim}
\usepackage{graphicx}
\usepackage{textcomp}
\usepackage{cases}
\usepackage{amsmath,amssymb,amsfonts}
\usepackage{fancyhdr,graphicx,amsmath,amssymb}
\usepackage[ruled,vlined]{algorithm2e}
\usepackage{xcolor}
\usepackage{subfigure}
\usepackage[utf8]{inputenc}
\usepackage{soul}
\usepackage[noend]{algpseudocode}
\usepackage{multirow}
\usepackage{fontenc}
\usepackage{amsthm}
\usepackage{enumitem}
\newtheorem{theorem}{Theorem}
\newtheorem{lemma}{Lemma}
\newtheorem{corollary}{Corollary}
\newtheorem{definition}{Definition}

\def\BibTeX{{\rm B\kern-.05em{\sc i\kern-.025em b}\kern-.08em
    T\kern-.1667em\lower.7ex\hbox{E}\kern-.125emX}}
\begin{document}

\title{A Theory of Second-Order Wireless Network Optimization and Its Application on AoI
{\footnotesize \textsuperscript{*}}
\thanks{This material is based upon work supported in part by NSF under Award Number ECCS-2127721, in part by the U.S. Army Research Laboratory and the U.S. Army Research Office under Grant Number W911NF-18-1-0331, and in part by Office of Naval Research under Contract N00014-21-1-2385.
}
}

\author{\IEEEauthorblockN{Daojing Guo,
Khaled Nakhleh, I-Hong Hou}
\IEEEauthorblockA{Dept. of ECE, Texas A\&M University} 
Email: \{daojing\_guo, khaled.jamal, ihou\}@tamu.edu 
\and
\IEEEauthorblockN{Sastry Kompella, Clement Kam}
\IEEEauthorblockA{Naval Research Laboratory} 
Email: \{sk, ckk\}@ieee.org
}

\maketitle
\begin{abstract}
This paper introduces a new theoretical framework for optimizing second-order behaviors of wireless networks. Unlike existing techniques for network utility maximization, which only considers first-order statistics, this framework models every random process by its mean and temporal variance. The inclusion of temporal variance makes this framework well-suited for modeling stateful fading wireless channels and emerging network performance metrics such as age-of-information (AoI). Using this framework, we sharply characterize the second-order capacity region of wireless access networks. We also propose a simple scheduling policy and prove that it can achieve every interior point in the second-order capacity region. To demonstrate the utility of this framework, we apply it for an important open problem: the optimization of AoI over Gilbert-Elliott channels. We show that this framework provides a very accurate characterization of AoI. Moreover, it leads to a tractable scheduling policy that outperforms other existing work.
\end{abstract}

\section{introduction}

There are two seemingly contradictory trends happening in the field of wireless network optimization. On one hand, the study of network utility maximization (NUM) has witnessed tremendous success in the past two decades. Techniques based on dual decomposition, Lyapunov function, etc., have been shown to produce tractable and optimal solutions in complex networks for a wide range of objectives, including maximizing spectrum efficiency, minimizing power consumption, enforcing fairness among clients, and the combination of these objectives. Recent studies have also established iterative algorithms that not only converge to the optimum, but also have provably fast convergence rate \cite{huang2014power, huang2009delay, liu2016heavy, chen2017learn, liu2016achieving}. On the other hand, there have been growing interests in new performance metrics for emerging network applications, such as quality-of-experience (QoE) for the application of video streaming and age-of-information (AoI) for the application of real-time state estimation. Surprisingly, except for a few special cases, the problem of optimizing these new performance metrics remain largely open. This raises the question: Why do existing NUM techniques fail to solve the optimization problem for these new performance metrics?

The fundamental reason is that current NUM techniques are only applicable to first-order performance metrics, while emerging new performance metrics involve higher-order behaviors. Existing NUM problems typically define the utility of a flow $n$ as $U_n(x_n)$, where $x_n$ is an asymptotic first-order performance metric, such as throughput (\emph{long-term average} number of packet deliveries per unit time), power consumption (\emph{long-term average} amount of energy consumption per unit time), and channel utilization (\emph{long-term average} number of transmissions per unit time). However, emerging performance metrics like QoE and AoI require the characterization of short-term network behaviors, and hence cannot be fully captured by asymptotic first-order statistics. 

To bridge the gap between NUM techniques and emerging performance metrics, we present a new framework of second-order wireless optimization. This framework consists of the second-order models, that is, the means and the temporal variances, of all random processes, including the channel qualities and packet deliveries of wireless clients. The incorporation of temporal variances enables this framework to better characterize stateful fading wireless channels, such as Gilbert-Elliott channels, and emerging performance metrics. 

Using this framework, we sharply characterize the second-order capacity region of wireless networks, which entails the set of means and temporal variances of packet deliveries that are feasible under the constraints of the second-order models of channel qualities. As a result, the problem of optimizing emerging performance metrics is reduced to one that finds the optimal means and temporal variances of packet deliveries within the second-order capacity region. We also propose a simple scheduling policy and show that it can achieve every interior point of the second-order capacity region.

To demonstrate the utility our framework, we apply it for an important open problem: Finding the optimal scheduling policy to minimize system-wide AoI over Gilbert-Elliott channels. We theoretically derive the closed-form expressions of the second-order models for Gilbert-Elliott channels. We also show that the AoI of each wireless client can be well-approximated by the mean and the temporal variance of its packet delivery process. We compare the system-wide AoI of our scheduling policy against other policies from recent studies on AoI minimization. Simulation results show that our policy achieves a smaller system-wide AoI. These results are especially significant when one considers that our policy is a generic second-order optimization policy, while the other policies are tailor-made to minimize the system-wide AoI. 

The rest of the paper is organized as follows: Section \ref{sec:model} formally defines the second-order models of channel qualities and packet deliveries and the problem of second-order optimization. Section \ref{sec: aoi} uses second-order models to formulate the problem of minimizing system-wide AoI over Gilbert-Elliott channels. Section \ref{sec: capacity region} derives an outer bound of the second-order capacity region. Section \ref{sec: scheduling policy} proposes a simple scheduling policy and shows that it achieves every interior point of the second-order capacity region. Section \ref{sec:simulation} presents our simulation results. Section \ref{sec:related} surveys some related studies. Finally, Section \ref{sec:conclusion} concludes the paper.
\section{System Model for Second-Order Wireless Network Optimization}
\label{sec:model}

We begin by describing a generic network optimization problem. Consider a wireless system where one AP serves $N$ clients, numbered as $\{1,2,\dots, N\}$. Time is slotted and denoted by $t = 1,2,3,\dots.$ We consider the ON-OFF channel model where the AP can schedule a client for transmission if and only if the channel for the client is ON. Let $X_i(t)$ be the indicator function that the channel for client $i$ is ON at time $t$. We assume that the sequence $\{X_i(1), X_i(2),\dots\}$ is governed by a stochastic positive-recurrent Markov process with finite states. In each time slot, if there is at least one client having an ON channel, then the AP selects a client with an ON channel and transmits a packet to it. Let $Z_i(t)$ be the indicator function that client $i$ receives a packet at time $t$. The empirical performance of client $i$ is modeled as a function of the entire sequence $\{Z_i(1), Z_i(2),\dots\}$. We note that the performance model is very general and covers virtually all existing network performance metrics, including both traditional ones like throughput and emerging ones like AoI. The network optimization problem is to find a scheduling policy that maximizes the total performance of the network.

Solving this generic network optimization problem is difficult because it requires solving an $N$-dimensional Markov decision process. As a result, except for a few special cases, there remains no tractable optimal solutions for many emerging network performance metrics like AoI. To circumvent this challenge, we propose capturing each random process by its second-order model, namely, its mean and temporal variance.

We first define the second-order model for channels. With a slight abuse of notations, let $X_S(t):=\max\{X_i(t)|i\in S\}$ be the indicator function that at least one client in $S$ has an ON channel at time $t$. Since all channels are governed by stochastic positive-recurrent Markov processes, the strong law of large numbers for Markov chains states that $\frac{\sum_{t=1}^TX_S(t)}{T}$ converges to a constant almost surely as $T\rightarrow\infty$. Hence, we can define the mean of $X_S$ as
\begin{equation}
    m_S:=\lim_{T\rightarrow\infty}\frac{\sum_{t=1}^TX_S(t)}{T}.
\end{equation}
The Markov central limit theorem further states that $\frac{\sum_{t=1}^TX_S(t)-Tm_S}{\sqrt{T}}$ converges in distribution to a Gaussian random variable as $T\rightarrow\infty$. Hence, we define the temporal variance of $X_S$ as
\begin{equation}
    v_S^2:=E[(\lim_{T\rightarrow\infty}\frac{\sum_{t=1}^TX_S(t)-Tm_S}{\sqrt{T}})^2].
\end{equation}
The second-order channel model is then expressed as the collection of the means and temporal variances of all $X_S$, namely, $\{(m_S, v_S^2)|S\subseteq\{1,2,\dots,N\}\}$.

The second-order model for packet deliveries is defined similarly. Assuming that the AP's scheduling policy is ergodic, we can define the mean and the temporal variance of $Z_i$ as
\begin{equation}
    \mu_i:=\lim_{T\rightarrow\infty}\frac{\sum_{t=1}^TZ_i(t)}{T}, \sigma_i^2:=E[(\lim_{T\rightarrow\infty}\frac{\sum_{t=1}^TZ_i(t)-T\mu_i}{\sqrt{T}})^2].
\end{equation}
The second-order delivery model is $\{(\mu_i,\sigma_i^2)|1\leq i \leq N\}$. The performance a client $i$ is modeled as a function of $(\mu_i,\sigma_i^2)$, which we denote by $F_i(\mu_i,\sigma_i^2)$.

Since clients want to have large means and small variances for their delivery processes, we define the second-order capacity region of a network as follows:
\begin{definition}
[Second-order capacity region] Given a second-order channel model $\{(m_S, v_S^2)|S\subseteq\{1,2,\dots,N\}\}$, the second-order capacity region is the set of all $\{(\mu_i,\sigma_i^2)|1\leq i \leq N\}$ such that there exists a scheduling policy under which $\lim_{T\rightarrow\infty}\frac{\sum_{t=1}^TZ_i(t)}{T}=\mu_i$ and $E[(\lim_{T\rightarrow\infty}\frac{\sum_{t=1}^TZ_i(t)-T\mu_i}{\sqrt{T}})^2]\leq\sigma_i^2,\forall i$. $\Box$
\end{definition}

The second-order network optimization problem entails finding the scheduling policy that maximizes $\sum_{i=1}^NF_i(\mu_i,\sigma_i^2)$.

\section{The Second-Order Model for AoI Optimization over Gilbert-Elliott Channels} \label{sec: aoi}

To demonstrate the utility of our second-order models, we derive the second-order models for an important, but unsolved, problem: the optimization of AoI over Gilbert-Elliott channels.

\subsection{The Second-Order Model of Gilbert-Elliott Channels}

\begin{figure}
\centering
\includegraphics[width=2in]{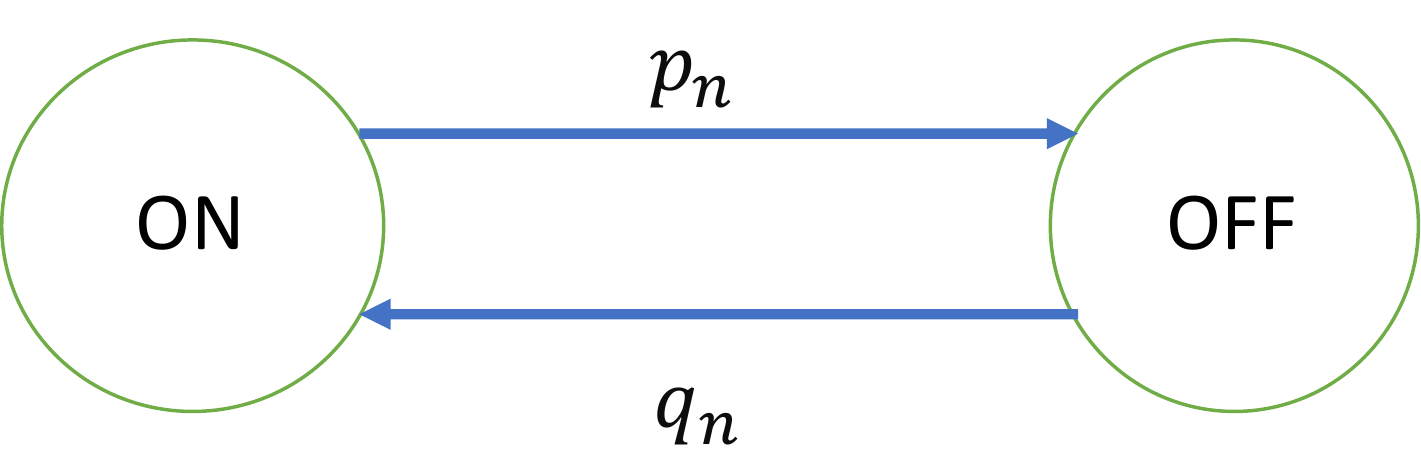}
\caption{The Gilbert-Elliott Model}\label{fig:GE} 
\vspace{-5pt}
\end{figure}

In Gilbert-Elliott channels \cite{Gilbert1960GEmodel,Elliot1963GEmodel}, the channel for each client $i$ is modeled as a two-state Markov process, as shown in Fig. ~\ref{fig:GE}. The channel is ON if it is in the good (G) state, and is OFF if it is in the bad (B) state. The transition probabilities from G to B and from B to G are $p_i$ and $q_i$, respectively. The channels are independent from each other.

We now show the second-order model of Gilbert-Elliott channels.
\begin{theorem} \label{theorem:second-order-GE}
Under the Gilbert-Elliott channels, for all $S$,
\begin{align}
    m_S=& 1-\prod_{i\in S}\frac{p_i}{p_i+q_i},\label{eq:GE mean}\\ 
    v_S^2=&2\sum_{k=1}^\infty\Big(\prod_{i\in S}G_i(k+1)-\prod_{i\in S}\frac{p_i}{p_i+q_i}\Big)\prod_{i\in S}\frac{p_i}{p_i+q_i}\nonumber\\
    &+\prod_{i\in S}\frac{p_i}{p_i+q_i}-(\prod_{i\in S}\frac{p_i}{p_i+q_i})^2,\label{eq:GE variance}
\end{align}
where $G_i(k)=\frac{p_i}{p_i+q_i}+\frac{q_i}{p_i+q_i}(1-p_i-q_i)^{k-1}$.
\end{theorem}
\begin{proof}
Let $Y_i(t):= 1- X_i(t)$ be the indicator function that client $i$ has an OFF channel at time $t$. Let $Y_S(t):=1-X_S(t)$ be the indicator function that all clients in the subset $S$ have OFF channels at time $t$. Hence, we have $Y_S(t)=\prod_{i\in S}Y_i(t)$. Suppose the Markov process of each channel is in the steady-state at time $t$, then we have $Prob(Y_i(t)=1)=\frac{p_i}{p_i+q_i}$. Hence, $E[Y_S(t)]=\prod_{i\in S}\frac{p_i}{p_i+q_i}$ and $E[X_S(t)] = 1-E[Y_S(t)]=1-\prod_{i\in S}\frac{p_i}{p_i+q_i}$. This establishes (\ref{eq:GE mean}).

Next, we establish (\ref{eq:GE variance}). We have $(\sum_{t=1}^TX_S(t)-Tm_S)^2=(\sum_{t=1}^TY_S(t)-T(1-m_S))^2$. By the Markov central limit theorem, we can calculate $v_S^2$ by assuming that the Markov process of each channel is in the steady-state at time $1$ and using the following formula:
\begin{equation}
    v_S^2=Var(Y_S(1))+2\sum_{k=1}^\infty Cov(Y_S(1), Y_S(1+k)).
\end{equation}
Since $Y_S(1)$ is a Bernoulli random variable with mean $\prod_{i\in S}\frac{p_i}{p_i+q_i}$, we have 
\begin{equation}Var(Y_S(1))=\prod_{i\in S}\frac{p_i}{p_i+q_i}-(\prod_{i\in S}\frac{p_i}{p_i+q_i})^2. \label{eq:Y_variance}
\end{equation}

Let $G_i(k)=Prob(Y_i(k)=1|Y_i(1)=1)$. Then,
\begin{align}
    &E[Y_S(1)Y_S(1+k)]\notag\\
    =&Prob(Y_S(1+k)=1|Y_S(1)=1)\times Prob(Y_S(1)=1)\notag\\
    =&Prob(Y_i(1+k)=1,\forall i\in S|Y_i(1)=1, \forall i\in S)\prod_{i\in S}\frac{p_i}{p_i+q_i}\notag\\
    =&\prod_{i\in S}G_i(k+1)\prod_{i\in S}\frac{p_i}{p_i+q_i},
\end{align}
and
\begin{align}
    &Cov(Y_S(1), Y_S(1+k))\notag\\
    =&E[Y_S(1)Y_S(1+k)]-E[Y_S(1)]E[Y_S(1+k)]\notag\\
    =&\Big(\prod_{i\in S}G_i(k+1)-\prod_{i\in S}\frac{p_i}{p_i+q_i}\Big)\prod_{i\in S}\frac{p_i}{p_i+q_i}\label{eq:Y_cov}
\end{align}
Combining (\ref{eq:Y_variance}) and (\ref{eq:Y_cov}) establishes (\ref{eq:GE variance}).

It remains to find the closed-form expression of $G_i(k)$. We have
\begin{align}
    &G_i(k)=Prob(Y_i(k)=1|Y_i(1)=1)\notag\\
    =&G_i(k-1)(1-q_i)+(1-G_i(k-1))p_i\notag\\
    =&p_i+(1-p_i-q_i)G_i(k-1),
\end{align}
if $k>1$, and $G_i(k)=1$, if $k=1$. Solving this recursive equation yields $G_i(k)=\frac{p_i}{p_i+q_i}+\frac{q_i}{p_i+q_i}(1-p_i-q_i)^{k-1}$. This completes the proof.
\end{proof}

When $p_i+q_i=1$, the Gilbert-Elliott channel reduces to the i.i.d. channel model where $X_i(t)=1$ with probability $q_i$, independent from any prior events. By replacing $p_i=1-q_i$, we obtain the second-order model of i.i.d. channels as below:
\begin{corollary}
Under the i.i.d. channels with $Prob(X_i(t)=1)=q_i$,
\begin{align}
    m_S=& 1-\prod_{i\in S}(1-q_i), 
    v_S^2=& \prod_{i\in S}(1-q_i)-\prod_{i\in S}(1-q_i)^2,
\end{align}
for all $S$. $\Box$
\end{corollary}

\subsection{The Second-Order Model of AoI Optimization}

Age-of-Information (AoI) has been proposed to model the performance of real-time remote sensing applications, where a controller is obtaining status updates from a number of sensors. In a nutshell, the AoI corresponding to a sensor at a given time is defined as the age of the newest information update that it has ever delivered to the controller. In terms of our network model, the AP is the controller and each client is a sensor.

Similar to the case studied in \cite{kadota2019minimizing}, we consider that each sensor $i$ generates new updates by a Bernoulli random process. In each time slot $t$, sensor $i$ generates a new update with probability $\lambda_i$, independent from any prior events. To minimize AoI, each sensor only keeps the most recent update in its memory, and it transmits the most recent update whenever it is scheduled for transmission. In other words, a sensor discards all its prior updates every time it generates a new update. The prior work \cite{kadota2019minimizing} considers that the controller knows when each sensor generates a new update. In practice, however, the controller cannot know whether a sensor has generated a new update until it schedules the sensor for transmission. In this paper, we further address the issue that the controller only knows $\lambda_i$ but not the exact times at which sensors generate new updates. Hence, we assume that the scheduling decision is independent from update generations.

Let $A_i(n):=\min\{\tau|\sum_{t=1}^\tau Z_i(t)=n\}$ be the time of the $n$-th delivery for client $i$, and let $B_i(n) := A_i(n+1)-A_i(n)$ be the time between the $n$-th and the $(n+1)$-th deliveries. Since scheduling decisions are independent from update generations, we have the following:
\begin{lemma}
If $\{B_i(0), B_i(1), \dots\}$ is independent from the update generation processes of sensor $i$, then the long-term average AoI of sensor $i$ is
\begin{equation}
    \overline{AoI}_i=\frac{E[B_i^2]}{2E[B_i]} + \frac{1}{\lambda_i}-\frac{1}{2},
\end{equation}
where $E[B_i^2]:=\lim_{m\rightarrow\infty}\sum_{n=1}^m B_i(n)^2/m$ and $E[B_i]:=\lim_{m\rightarrow\infty}\sum_{n=1}^m B_i(n)/m$.
\end{lemma}
\begin{proof}
This lemma can be established by combining techniques in the proof of Proposition 2 in \cite{kadota2019minimizing} and the fact that $B_i(n)$ is independent from update generations. The complete proof is omitted due to space limitation.
\end{proof}

We aim to express $\overline{AoI}_i$ as a function of the second-order delivery model of client $i$, $(\mu_i, \sigma_i^2)$. Since there can be multiple sequences of $\{Z_i(1), Z_i(2),\dots\}$ with the same $(\mu_i, \sigma_i^2)$, we will derive $\overline{AoI}_i$ with respect to a \emph{second-order reference delivery process} as defined below.

Let $BM_{\mu_i,\sigma_i^2}(t)$ be a Brownian motion random process with mean $\mu_i$ and variance $\sigma_i^2$. An important property of the Brownian motion random process is that for any $t_1<t_2$, $BM_{\mu_i,\sigma_i^2}(t_1)-BM_{\mu_i,\sigma_i^2}(t_2)$ is a Gaussian random variable with mean $(t_2-t_1)\mu_i$ and variance $(t_2-t_1)\sigma_i^2$. Our goal is to define a sequence $\{Z'_i(1), Z'_i(2),\dots\}$ such that $\sum_{\tau=1}^t Z'_i(\tau)\approx BM_{\mu_i,\sigma_i^2}(t)$.

\begin{definition}
Given $(\mu_i, \sigma_i^2)$, the second-order reference delivery process, denoted by $\{Z'_i(1), Z'_i(2),\dots\}$ is defined to be
\begin{equation}
    Z'_i(t)=\left\{\begin{array}{ll}1&\mbox{if $BM_{\mu_i,\sigma_i^2}(t)-BM_{\mu_i,\sigma_i^2}(t^-)\geq 1$,}\\
    0&\mbox{else,}
    \end{array}
    \right.\label{eq:Z'def}
\end{equation}
where $t^-:=\max\{\tau|\tau<t, Z'_i(\tau)=1\}$. $\Box$
\end{definition}

We now derive $\overline{AoI}_i$ with respect to the sequence $\{Z'_i(1), Z'_i(2),\dots\}$. Consider the time between the $n$-th and the $(n+1)$-th deliveries, which is denoted by $B_i(n)$, under the sequence $\{Z'_i(1), Z'_i(2),\dots\}$. From (\ref{eq:Z'def}), $B_i(n)$ can be approximated by the amount of time needed for the Brownian motion random process to increase by 1, which is equivalent to the first-hitting time for a fixed level 1 and we denote it by $H_i$. It has been shown that the the first-hitting time for a fixed level 1 follows the inverse Gaussian distribution $IG(\frac{1}{\mu_i}, \frac{1}{\sigma_i^2})$ \cite{Schrodinger1915, folks1978inverse}. Hence, we have $E[H_i]=1/\mu_i$ and $E[H_i^2]=\sigma_i^2/\mu_i^3+1/\mu_i^2$. We now have
\begin{align}
    &\overline{AoI}_i=\frac{E[B_i^2]}{2E[B_i]}+\frac{1}{\lambda_i}-\frac{1}{2}\nonumber\\
    \approx&\frac{E[H_i^2]}{2E[H_i]}+\frac{1}{\lambda_i}-\frac{1}{2}=\frac{1}{2}(\frac{\sigma_i^2}{\mu_i^2}+\frac{1}{\mu_i}) +\frac{1}{\lambda_i}-\frac{1}{2}. \label{eq:AoI approx}
\end{align}

\subsection{Model Validation}

We now verify whether the second-order model provides a good approximation of AoI over Gilbert-Elliott channels. We consider a system with only one client (sensor). The AP (controller) schedules the client for transmission whenever the client has an ON channel. Hence, we have $\mu_1=m_{\{1\}}$ and $\sigma_1^2=v_{\{1\}}^2$. Given, $p_1$, $q_1$, and $\lambda_1$, we can combine (\ref{eq:GE mean}), (\ref{eq:GE variance}), and (\ref{eq:AoI approx}) to obtain a theoretical approximation of the AoI. We note that (\ref{eq:GE variance}) involves a summation of infinite terms $\sum_{k=1}^\infty (G_1(k)-\frac{p_1}{p_1+q_1})$. Since $G_1(k)$ converges to $\frac{p_1}{p_1+q_1}$ exponentially fast, we replace this term with $\sum_{k=1}^{100} (G_1(k)-\frac{p_1}{p_1+q_1})$ when calculating $v_{\{1\}}^2$.

We evaluate the accuracy of the theoretical AoI over a wide range of $(p_1, q_1, \lambda_1)$. For each $(p_1, q_1, \lambda_1)$, we obtain the empirical AoI by simulation the system for 1000 runs, where each run contains 50,000 time slots. The results are shown in Fig. \ref{fig:single_client_validation}. It can be observed that the theoretical AoI is always almost identical to the empirical AoI under all settings. The largest difference between theoretical and empirical AoI among all evaluated case is only 0.00558. 

\begin{figure}[ht]

\begin{center} 
\subfigure[$q=0.2$. $\lambda = 1$.]{\includegraphics[width=1.7in, height=1.37in]{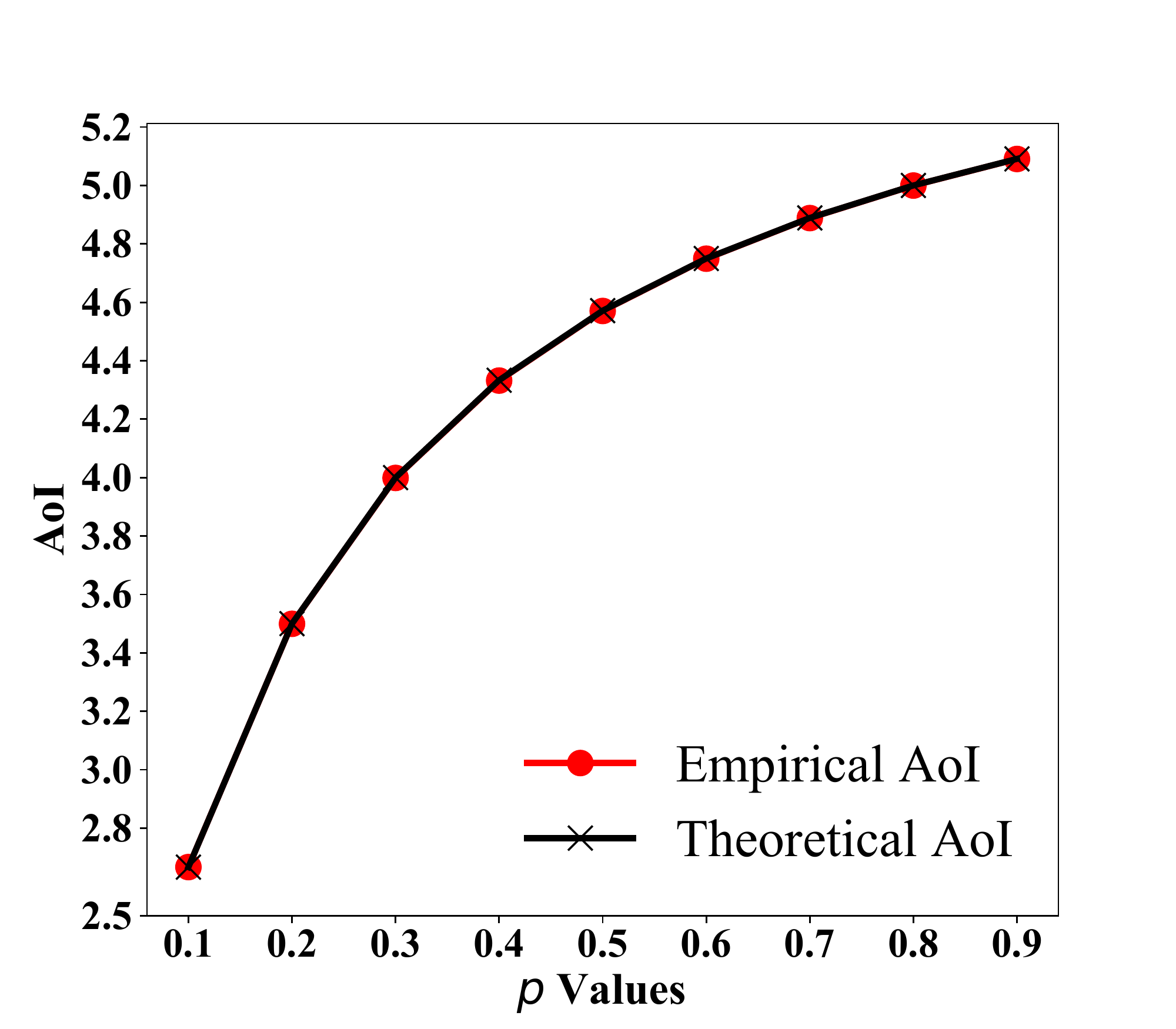}}
\subfigure[$q=0.2$. $\lambda = 0.1$.]{\includegraphics[width=1.7in, height=1.37in]{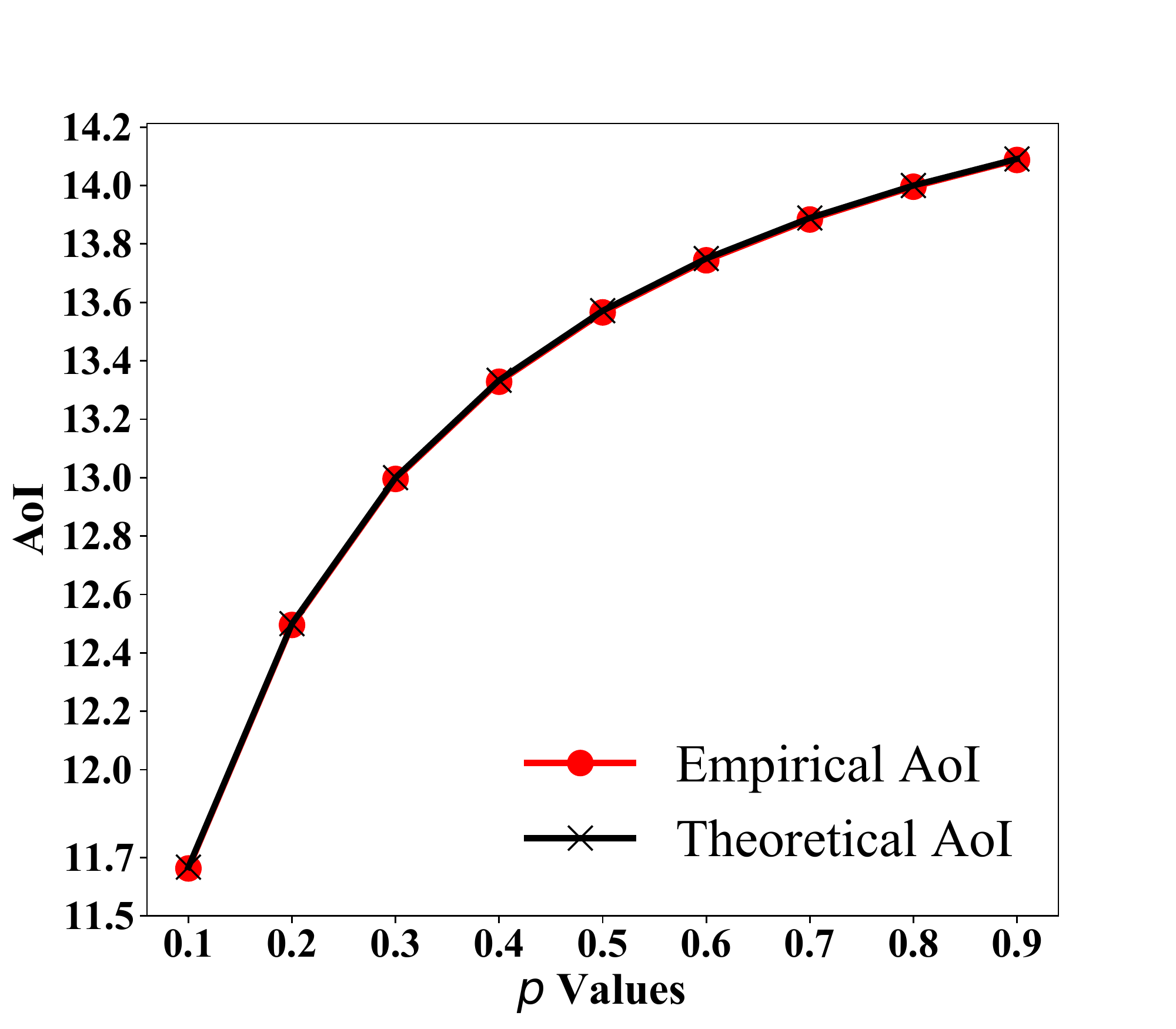}}
\subfigure[$q=0.8$. $\lambda = 1$.]{\includegraphics[width=1.7in, height=1.37in]{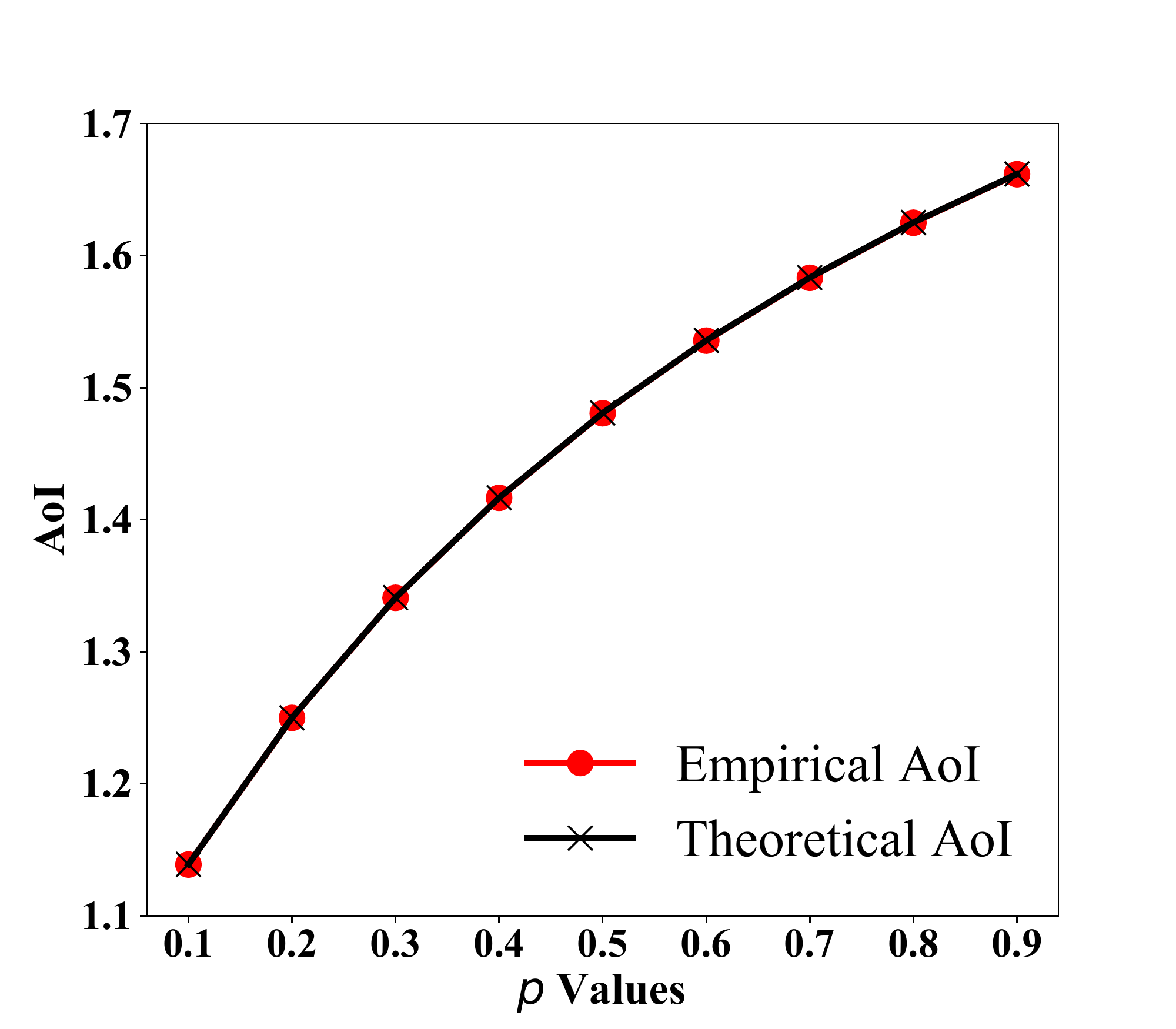}}
\subfigure[$q=0.8$. $\lambda = 0.1$.]{\includegraphics[width=1.7in, 
height=1.37in]{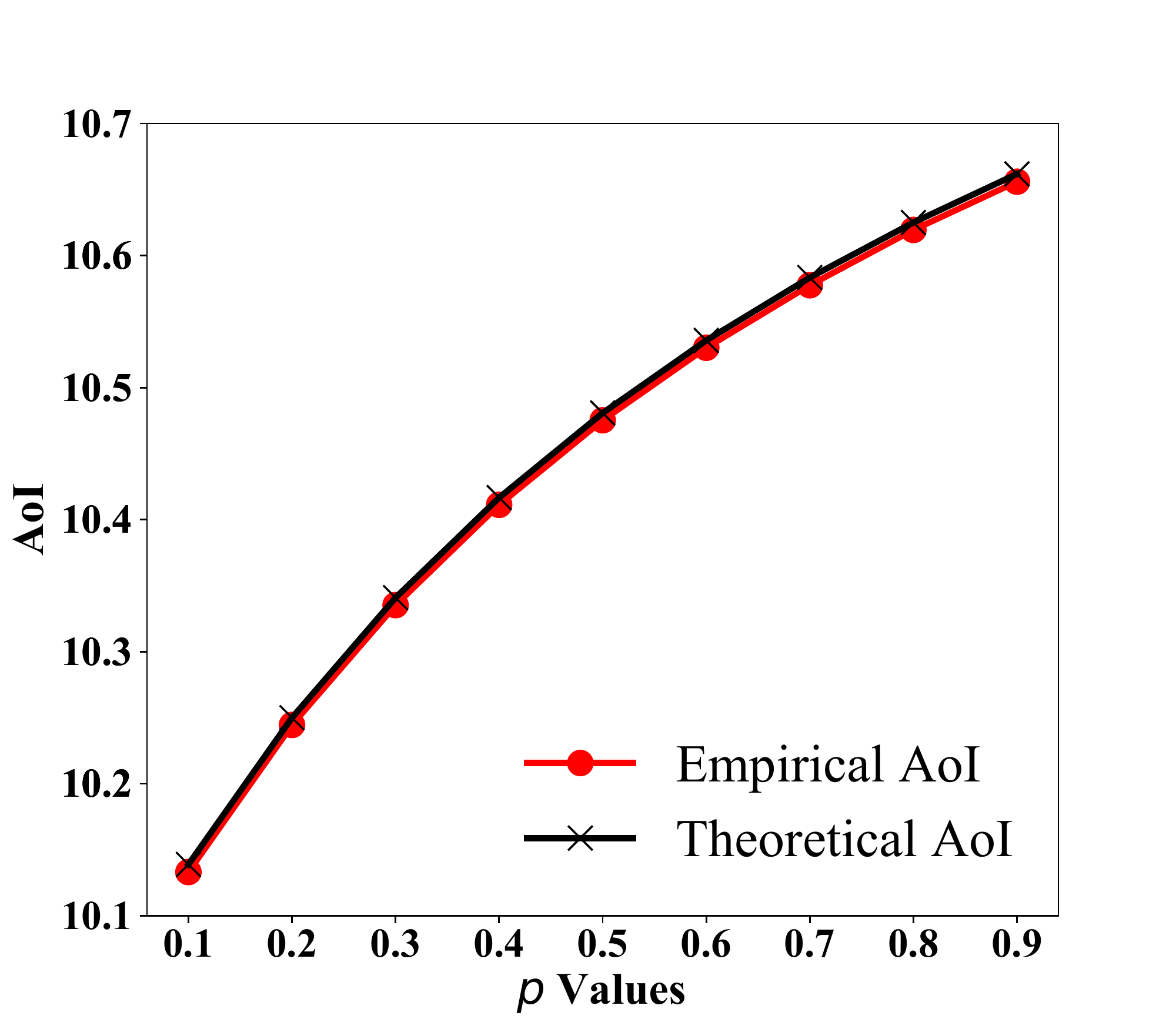}}
\end{center}
\caption{Model Validation For A Single Client.}
\label{fig:single_client_validation}
\end{figure}

\section{An Outer Bound of the Second-Order Capacity Region} \label{sec: capacity region}

In this section, we derive a necessary condition for the second-order delivery model $\{(\mu_i,\sigma_i^2)|1\leq i\leq N\}$ to be in the second-order capacity region.

\begin{theorem}\label{theorem:outer bound}
Given a second-order channel model $\{(m_S, v_S^2)|S\subseteq\{1,2,\dots,N\}\}$, a second-order delivery model $\{(\mu_i,\sigma_i^2)|1\leq i\leq N\}$ can be in the second-order capacity region only if
\begin{align}
    &\sum_{i\in S}\mu_i\leq m_S, \forall S\subseteq\{1,2,\dots,N\},\label{eq:necessary:mean}\\
    &\sum_{i=1}^N\mu_i=m_{\{1,2,\dots,N\}},\label{eq:necessary:total mean}\\
    &\sum_{i=1}^{N}\sqrt{\sigma_i^2}\geq \sqrt{v_{\{1,2,\dots,N\}}^2},\label{eq:necessary:variance}\\
    &\mu_i\geq 0, \forall i.\label{eq:necessary:non-negative}
\end{align}
\end{theorem}
\begin{proof}
We first establish (\ref{eq:necessary:mean}). The AP can transmit a packet to a client $i$ at time $t$ only if the client has an ON channel, that is, $X_i(t)=1$. Moreover, the AP can transmit to at most one client in each time slot. Hence, we have $\sum_{i\in S}Z_i(t)\leq X_S(t)$ under any scheduling policy. This gives us
\begin{align}
    &\sum_{i\in S}\mu_i=\lim_{T\rightarrow\infty}\frac{\sum_{i\in S}\sum_{t=1}^TZ_i(t)}{T}\nonumber\\
    \leq&\lim_{T\rightarrow\infty}\frac{\sum_{t=1}^TX_S(t)}{T}=m_S, \forall S\subseteq\{1,2,\dots, N\}.
\end{align}

We can similarly establish (\ref{eq:necessary:total mean}) by noting that $\sum_{i=1}^NZ_i(t)= X_{\{1,2,\dots,N\}}(t)$, since the AP always transmits one packet as long as at least one client has an ON channel.

Finally, we establish (\ref{eq:necessary:variance}). Let $\hat{X}_S$ be the random variable $\lim_{T\rightarrow\infty}\frac{\sum_{t=1}^TX_S(t)-Tm_S}{\sqrt{T}}$ and $\hat{Z}_i$ be the random variable $\lim_{T\rightarrow\infty}\frac{\sum_{t=1}^TZ_i(t)-T\mu_i}{\sqrt{T}}$. Since $\sum_{i=1}^NZ_i(t)= X_{\{1,2,\dots,N\}}(t)$ and (\ref{eq:necessary:total mean}), we have $\sum_{i=1}^N\hat{Z}_i=\hat{X}_{\{1,2,\dots,N\}}$. We then have
\begin{align}
    &(\sum_{i=1}^N\sqrt{\sigma_i^2})^2=(\sum_{i=1}^N\sqrt{E[\hat{Z}_i^2]})^2\nonumber\\
    =&\sum_{i=1}^NE[\hat{Z}_i^2]+2\sum_{i\neq j}\sqrt{E[\hat{Z}_i^2]E[\hat{Z}_j^2]}\nonumber\\
    \geq&\sum_{i=1}^NE[\hat{Z}_i^2]+2\sum_{i\neq j}E[\hat{Z}_i\hat{Z}_j]\quad(\mbox{Cauchy-Schwarz inequality})\nonumber\\
    =&E[(\sum_{i=1}^N\hat{Z}_i)^2]=E[\hat{X}_{\{1,2,\dots,N\}}^2]=v_{\{1,2,\dots,N\}}^2.
\end{align}
This completes the proof.
\end{proof}

\section{Scheduling Policy with Tight Inner Bound}\label{sec: scheduling policy}

In this section, we derive a sufficient condition for the second-order delivery model $\{(\mu_i,\sigma_i^2)|1\leq i\leq N\}$ to be in the second-order capacity region. We also propose a simple scheduling policy that delivers the desirable second-order delivery models as long as they satisfy the sufficient condition. We state the sufficient condition as follows:

\begin{theorem} \label{theorem:inner bound}
Given a second-order channel model $\{(m_S, v_S^2)|S\subseteq\{1,2,\dots,N\}\}$, a second-order delivery model $\{(\mu_i,\sigma_i^2)|1\leq i\leq N\}$ is in the second-order capacity region if

\begin{align}
    &\sum_{i\in S}\mu_i< m_S, \forall S\subsetneq\{1,2,\dots,N\},\label{eq:sufficient:mean}\\
    &\sum_{i=1}^N\mu_i=m_{\{1,2,\dots,N\}},\label{eq:sufficient:total mean}\\
    &\sum_{i=1}^{N}\sqrt{\sigma_i^2}\geq \sqrt{v_{\{1,2,\dots,N\}}^2},\label{eq:sufficient:variance}\\
    &\mu_i\geq 0, \sigma_i^2>0 \forall i.\label{eq:sufficient:non-negative}
\end{align}
$
\Box
$
\end{theorem}

Before proving Theorem~\ref{theorem:inner bound}, we first discuss its implications. Comparing the conditions in Theorems~\ref{theorem:outer bound} and \ref{theorem:inner bound}, we note that the only difference is that the sufficient condition requires strict inequality for (\ref{eq:necessary:mean}) for all proper subsets. Hence, the sufficient condition describes an inner bound that is almost tight except on some boundaries.

We prove Theorem~\ref{theorem:inner bound} by proposing a scheduling that achieves every point in the inner bound. Given $\{(\mu_i,\sigma_i^2)|1\leq i\leq N\}$, define the \emph{deficit} of a client $i$ at time $t$ as $d_i(t) = t\mu_i-\sum_{\tau=1}^tZ_i(\tau)$. In each time slot $t$, the AP chooses the client with the largest $d_i(t-1)/\sqrt{\sigma_i^2}$ among those with ON channels and transmits a packet to the chosen client. We call this scheduling policy the \emph{variance-weighted-deficit} (VWD) policy. 

We now analyze the performance of the VWD policy. Let $D(t):=\sum_{i=1}^Nd_i(t)/\sum_{i=1}^N\sqrt{\sigma_i^2}$. We then have 
\begin{align}
    &\Delta d_i(t):=d_i(t)-d_i(t-1)=\mu_i-Z_i(t),\\
    &\Delta D(t):=D(t)-D(t-1)\nonumber\\
    =&\frac{\sum_{i=1}^N\mu_i-\sum_{i=1}^NZ_i(t)}{\sum_{i=1}^N\sqrt{\sigma_i^2}}=\frac{m_{\{1,2,\dots,N\}}-X_{\{1,2,\dots,N\}}(t)}{\sum_{i=1}^N\sqrt{\sigma_i^2}}.
\end{align}

Consider the Lyapunov function $L(t):=\frac{1}{2}\sum_{i=1}^N\sqrt{\sigma_i^2}\Big(\frac{d_i(t)}{\sqrt{\sigma_i^2}}-D(t)\Big)^2$. Let $H^t$ be the system history up to time $t$. We can derive the expected one-step Lyapunov drift as
\begin{align}
    &\Delta (L(t)) := E[L(t) - L(t-1)|H^{t-1}]\notag \\ 
    =& E[\frac{1}{2} \sum_{i=1}^N \sqrt{\sigma_i^2}\Big(\frac{d_i(t)}{\sqrt{\sigma_i^2}}-D(t))\Big)^2\nonumber\\
    &-\frac{1}{2} \sum_{i=1}^N\sqrt{\sigma_i^2}\Big(\frac{d_i(t-1)}{\sqrt{\sigma_i^2}}-D(t-1)\Big)^2|H^{t-1}]\nonumber\\
    =&E[\sum_{i=1}^N\sqrt{\sigma_i^2}\Big(\frac{d_i(t-1)}{\sqrt{\sigma_i^2}}-D(t-1)\Big)\Big(\frac{\Delta d_i(t)}{\sqrt{\sigma_i^2}}-\Delta D(t)\Big)\nonumber\\
    &+\frac{1}{2} \sum_{i=1}^N \sqrt{\sigma_i^2}\Big(\frac{\Delta d_i(t)}{\sqrt{\sigma_i^2}}-\Delta D(t))\Big)^2|H^{t-1}]\nonumber\\
    \leq &B + E[\sum_{i=1}^N\Big(\frac{d_i(t-1)}{\sqrt{\sigma_i^2}}-D(t-1)\Big)\Delta d_i(t)\nonumber\\
    &-\sum_{i=1}^N\sqrt{\sigma_i^2}\Big(\frac{d_i(t-1)}{\sqrt{\sigma_i^2}}-D(t-1)\Big)\Delta D(t)|H^{t-1}]\nonumber\\
    =&B + E[\sum_{i=1}^N\Big(\frac{d_i(t-1)}{\sqrt{\sigma_i^2}}-D(t-1)\Big)\Delta d_i(t)|H^{t-1}], \label{eq:one-step-Lyapunov}
\end{align}
where $B$ is a bounded constant. The last two steps follow because $\Delta d_i(t)$ and $\Delta D(t)$ are bounded and because $\sum_{i=1}^Nd_i(t-1)=\sum_{i=1}^N\sqrt{\sigma_i^2}D(t-1)$.

The VWD policy schedules the client with the largest $d_i(t-1)/\sqrt{\sigma_i^2}$, which is also the client with the largest $d_i(t-1)/\sqrt{\sigma_i^2} - D(t-1)$, among those with ON channels. Hence, under the VWD policy, the system can be modeled as a Markov process whose state consists of the channel states and $d_i(t-1)/\sqrt{\sigma_i^2} - D(t-1)$ of all clients. Further, the VWD policy is the policy that minimizes $E[\sum_{i=1}^N\Big(\frac{d_i(t-1)}{\sqrt{\sigma_i^2}}-D(t-1)\Big)\Delta d_i(t)|H^{t-1}]$ for all $t$. We first show that the Markov process is positive-recurrent.

\begin{lemma}
Assume that (\ref{eq:sufficient:mean}) -- (\ref{eq:sufficient:non-negative}) are satisfied. Then, under the VWD policy, the system-wide Markov process, whose state consists of the channel states and $d_i(t-1)/\sqrt{\sigma_i^2} - D(t-1)$ of all clients, is positive-recurrent.
\end{lemma}
\begin{proof}
Due to (\ref{eq:sufficient:mean}), we can define 
\begin{equation}
    \delta := \min\{m_S-\sum_{i\in S}\mu_i|S\subsetneq\{1,2,\dots,N\}\}>0. \label{eq:delta}
\end{equation}
Further, since the channel of each client follows a positive-recurrent Markov process with finite states, there exists a finite number $\mathbb{T}$ such that
\begin{equation}
    \mathbb{T}m_S-\frac{\delta}{2}\leq E[\sum_{t=\tau+1}^{\tau+\mathbb{T}}X_S(t)|H^{\tau}]\leq \mathbb{T}m_S+\frac{\delta}{2}, \label{eq:T-steps-mean}
\end{equation}
for any $H^{\tau}$.

Let $L^V(t)$ and $\Delta d_i^V(t)$ be the values of $L(t)$ and $d_i(t)$ under the VWD policy. From (\ref{eq:one-step-Lyapunov}), we can bound the $\mathbb{T}$-step Lyapunov drift by
\begin{align}
    &E[L^V(\tau+\mathbb{T})-L^V(\tau)|H^\tau]\nonumber\\
    \leq &B\mathbb{T}+E[\sum_{t=\tau+1}^{\tau+\mathbb{T}}\sum_{i=1}^N\Big(\frac{d_i(t-1)}{\sqrt{\sigma_i^2}}-D(t-1)\Big)\Delta d_i^V(t)|H^{\tau}]\nonumber\\
    \leq &B\mathbb{T}+E[\sum_{t=\tau+1}^{\tau+\mathbb{T}}\sum_{i=1}^N\Big(\frac{d_i(t-1)}{\sqrt{\sigma_i^2}}-D(t-1)\Big)\Delta d_i^\eta(t)|H^{\tau}]\nonumber\\
    \leq &A+E[\sum_{i=1}^N\Big(\frac{d_i(\tau)}{\sqrt{\sigma_i^2}}-D(\tau)\Big)(\sum_{t=\tau+1}^{\tau+\mathbb{T}}\Delta d_i^\eta(t))|H^{\tau}], \label{eq:t steps drift}
\end{align}
for any other scheduling policy $\eta$, where $d_i^\eta(t)$ is the value of $d_i(t)$ under $\eta$ and $A$ is a bounded constant. The last inequality follows because $\mathbb{T}$, $|d_i(t)-d_i(\tau)|$, and $\Delta d_i(t)$ are all bounded for all $t\in [\tau+1, \tau+\mathbb{T}]$.

We now consider the scheduling policy $\eta$ that schedules the flow with the largest $d_i(\tau)/\sqrt{\sigma_i^2}$ among those with ON channels in all time slots $t\in [\tau+1, \tau+\mathbb{T}]$.

Without loss of generality, we assume that $d_1(\tau)/\sqrt{\sigma_1^2}\geq d_2(\tau)/\sqrt{\sigma_2^2}\geq\dots$. Under $\eta$, a client $i$ will be scheduled in time slot $t$ if it has an ON channel and all clients in $\{1,2,\dots, i-1\}$ have OFF channels, that is, $X_{\{1,2,\dots i\}}(t)=1$ and $X_{\{1,2,\dots i-1\}}(t)=0$. We hence have $\sum_{t=\tau+1}^{\tau+\mathbb{T}}Z_i(t)=\sum_{t=\tau+1}^{\tau+\mathbb{T}}X_{\{1,2,\dots i\}}(t)-\sum_{t=\tau+1}^{\tau+\mathbb{T}}X_{\{1,2,\dots i-1\}}(t)$. Therefore,
\begin{align}
    &E[\sum_{i=1}^N\Big(\frac{d_i(\tau)}{\sqrt{\sigma_i^2}}-D(\tau)\Big)(\sum_{t=\tau+1}^{\tau+\mathbb{T}}\Delta d_i^\eta(t))|H^{\tau}]\nonumber\\
    =&E[\sum_{i=1}^{N-1}\Big(\frac{d_i(\tau)}{\sqrt{\sigma_i^2}}-\frac{d_{i+1}(\tau)}{\sqrt{\sigma_{i+1}^2}}\Big)(\mathbb{T}\sum_{j=1}^i\mu_j\nonumber\\
    &-\sum_{t=\tau+1}^{\tau+\mathbb{T}}X_{\{1,2,\dots,i\}}(t))+\Big(\frac{d_N(\tau)}{\sqrt{\sigma_N^2}}-D(\tau)\Big)\nonumber\\
    &\times(\mathbb{T}\sum_{j=1}^N\mu_j-\sum_{t=\tau+1}^{\tau+\mathbb{T}}X_{\{1,2,\dots,N\}}(t))|H^{\tau}]\nonumber\\
    \leq& \sum_{i=1}^{N-1}\Big(\frac{d_i(\tau)}{\sqrt{\sigma_i^2}}-\frac{d_{i+1}(\tau)}{\sqrt{\sigma_{i+1}^2}}\Big)(-\delta/2)
    +\Big(\frac{d_N(\tau)}{\sqrt{\sigma_N^2}}-D(\tau)\Big)(-\delta/2)\nonumber\\
    =&\Big(\frac{d_1(\tau)}{\sqrt{\sigma_1^2}}-D(\tau)\Big)(-\delta/2), \label{eq:eta drift}
\end{align}
where the inequality holds due to (\ref{eq:sufficient:total mean}), (\ref{eq:delta}), and (\ref{eq:T-steps-mean}).

Combining (\ref{eq:t steps drift}) and (\ref{eq:eta drift}), and we have
\begin{equation}
    E[L^V(\tau+\mathbb{T})-L^V(\tau)|H^\tau]<-\delta,
\end{equation}
if $\max_i\Big(\frac{d_i(\tau)}{\sqrt{\sigma_i^2}}-D(\tau)\Big)>2(A/\delta+1)$, and
\begin{equation}
    E[L^V(\tau+\mathbb{T})-L^V(\tau)|H^\tau]\leq A,
\end{equation}
if $\max_i\Big(\frac{d_i(\tau)}{\sqrt{\sigma_i^2}}-D(\tau)\Big)\leq 2(A/\delta+1)$. Recall that $\sum_i \Big(\frac{d_i(\tau-1)}{\sqrt{\sigma_i^2}}-D(\tau-1)\Big)=0$ and the channel of each client follows a Markov process with finite states. Hence, all states of the system with $\max_i\Big(\frac{d_i(\tau)}{\sqrt{\sigma_i^2}}-D(\tau)\Big)\leq 2(A/\delta+1)$ belong to a finite set of states. By the Foster-Lyapunov Theorem, the system-wide Markov process is positive-recurrent.
\end{proof}

We now show that the VWD policy delivers all desirable second-order delivery models that satisfy the sufficient conditions (\ref{eq:sufficient:mean}) -- (\ref{eq:sufficient:non-negative}), and thereby establishing Theorem~\ref{theorem:inner bound}.

\begin{theorem}\label{thm: approach solution}
Assume that (\ref{eq:sufficient:mean}) -- (\ref{eq:sufficient:non-negative}) are satisfied. Then, under the VWD policy, $\lim_{T\rightarrow\infty}\frac{\sum_{t=1}^TZ_i(t)}{T}=\mu_i$ and $E[(\lim_{T\rightarrow\infty}\frac{\sum_{t=1}^TZ_i(t)-T\mu_i}{\sqrt{T}})^2]\leq\sigma_i^2,\forall i$.
\end{theorem} 
\begin{proof}

Since the system-wide Markov process is positive recurrent under the VWD policy, we have: \begin{align}
    \lim_{T\to\infty} \frac{d_i(T)/\sqrt{\sigma_i^2}-D(T)}{T} \to 0, \forall i, \label{f:positive recurrent for mean} \\
    \lim_{T\to\infty} \frac{d_i(T)/\sqrt{\sigma_i^2}-D(T)}{\sqrt{T}} \to 0, \forall i. \label{f:positive recurrent for var}
\end{align}

First, we show that $\lim_{T\rightarrow\infty}\frac{\sum_{t=1}^TZ_i(t)}{T}=\mu_i,\forall i$.
Recall that $d_i(t) = t\mu_i - \sum_{\tau=1}^tZ_i(\tau)$ and $D(t) = {\sum_{i=1}^N d_i(t)}/{\sum_{i=1}^N \sqrt{\sigma_i^2}}$.
By (\ref{eq:sufficient:total mean}), we have:
\begin{align}
     &\lim_{T \to \infty} \frac{D(T)}{T} = \lim_{T \to \infty} \frac{\sum_{i=1}^N T\mu_i -\sum_{t=1}^T\sum_{i=1}^N  Z_i(t)}{T\sum_{i=1}^N \sqrt{\sigma_i^2}} \notag \\
     =&\lim_{T\to \infty} \frac{Tm_{\{1,2,\dots,N\}} - \sum_{t=1}^TX_{\{1,2,\dots,N\}}}{{T\sum_{i=1}^N \sqrt{\sigma_i^2}}} = 0. \label{f:proving mean converge for D}
\end{align}
Hence, by (\ref{f:positive recurrent for mean}), we have $\lim_{T\to\infty} \frac{d_i(T)}{T}=\mu_i-\lim_{T\to\infty} \frac{\sum_{t=1}^TZ_i(t)}{T}=0$, for all $i$.

Next, we show that $E[(\lim_{T\rightarrow\infty}\frac{\sum_{t=1}^TZ_i(t)-T\mu_i}{\sqrt{T}})^2]\leq\sigma_i^2,\forall i$. We have, by (\ref{eq:sufficient:variance}),
\begin{align}
    E[(\lim_{T\rightarrow\infty}\frac{D(T)}{\sqrt{T}})^2]=\frac{v^2_{\{1,2,\dots,N\}}}{(\sum_{i=1}^N\sqrt{\sigma_i^2})^2}\leq 1,
\end{align}

and, hence,
\begin{align}
    &E[(\lim_{T\rightarrow\infty}\frac{\sum_{t=1}^TZ_i(t)-T\mu_i}{\sqrt{T}})^2] = E[(\lim_{T\rightarrow\infty}\frac{d_i(T)}{\sqrt{T}})^2]\notag\\
    =&\sigma_i^2E[(\lim_{T\rightarrow\infty}\frac{D(T)}{\sqrt{T}})^2]\leq \sigma_i^2.
\end{align}
\end{proof}

\begin{figure*}[hbt!]
\begin{center} 
\subfigure[N = 5 Clients.]{\includegraphics[width=0.3\textwidth, height=1.8in]{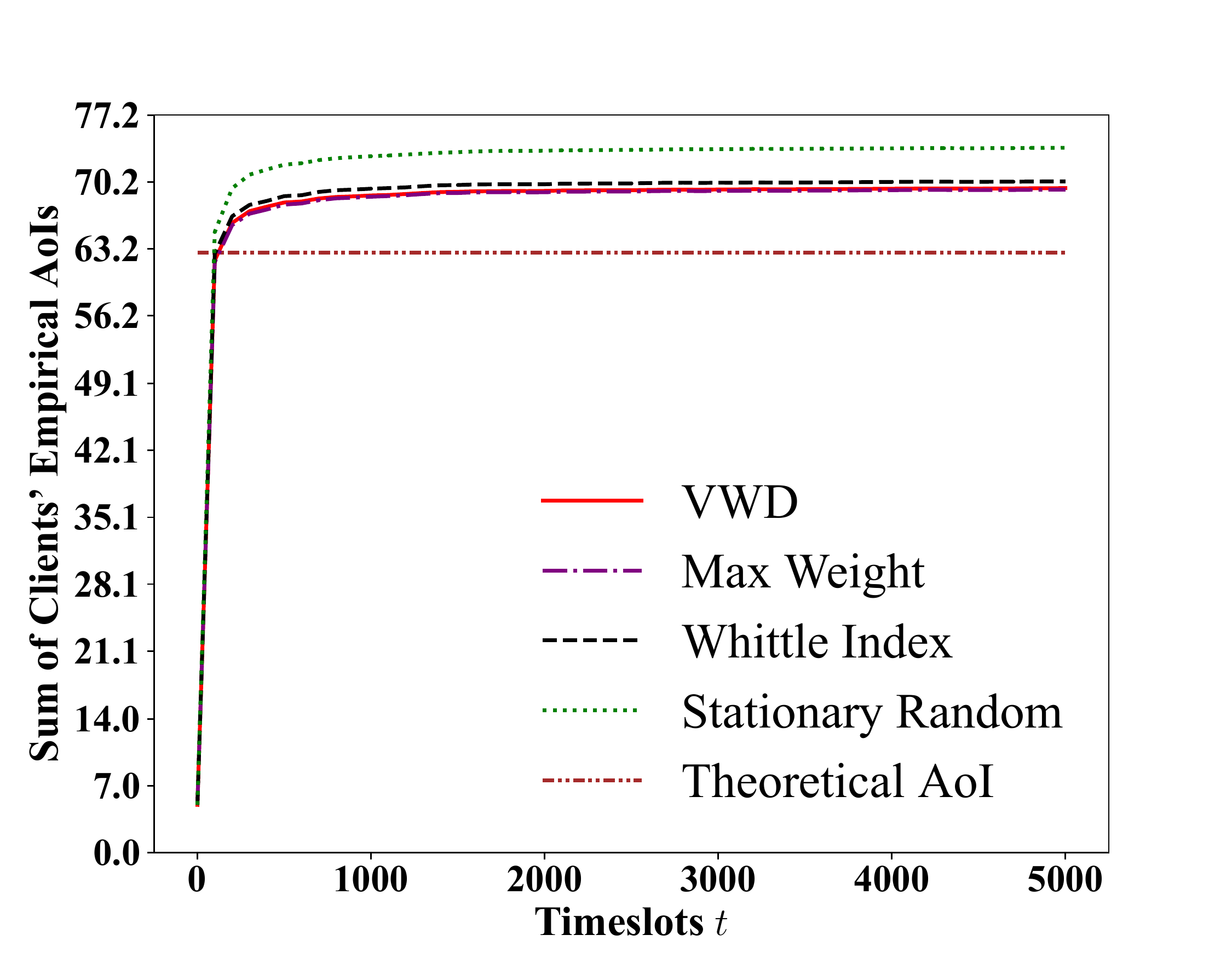}}
\subfigure[N = 10 Clients.]{\includegraphics[width=0.3\textwidth, height=1.8in]{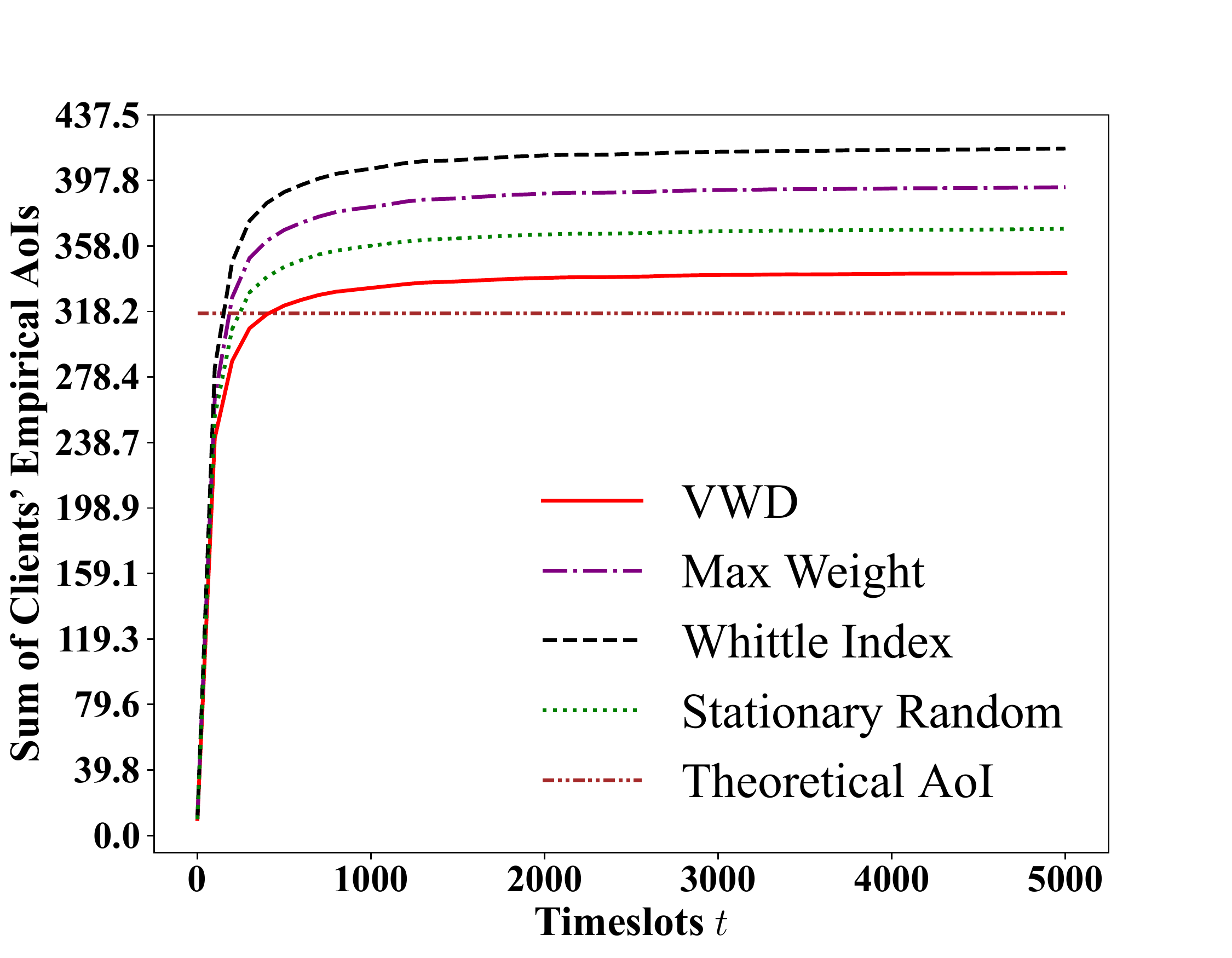}}
\subfigure[N = 20 Clients.]{\includegraphics[width=0.3\textwidth, height=1.8in]{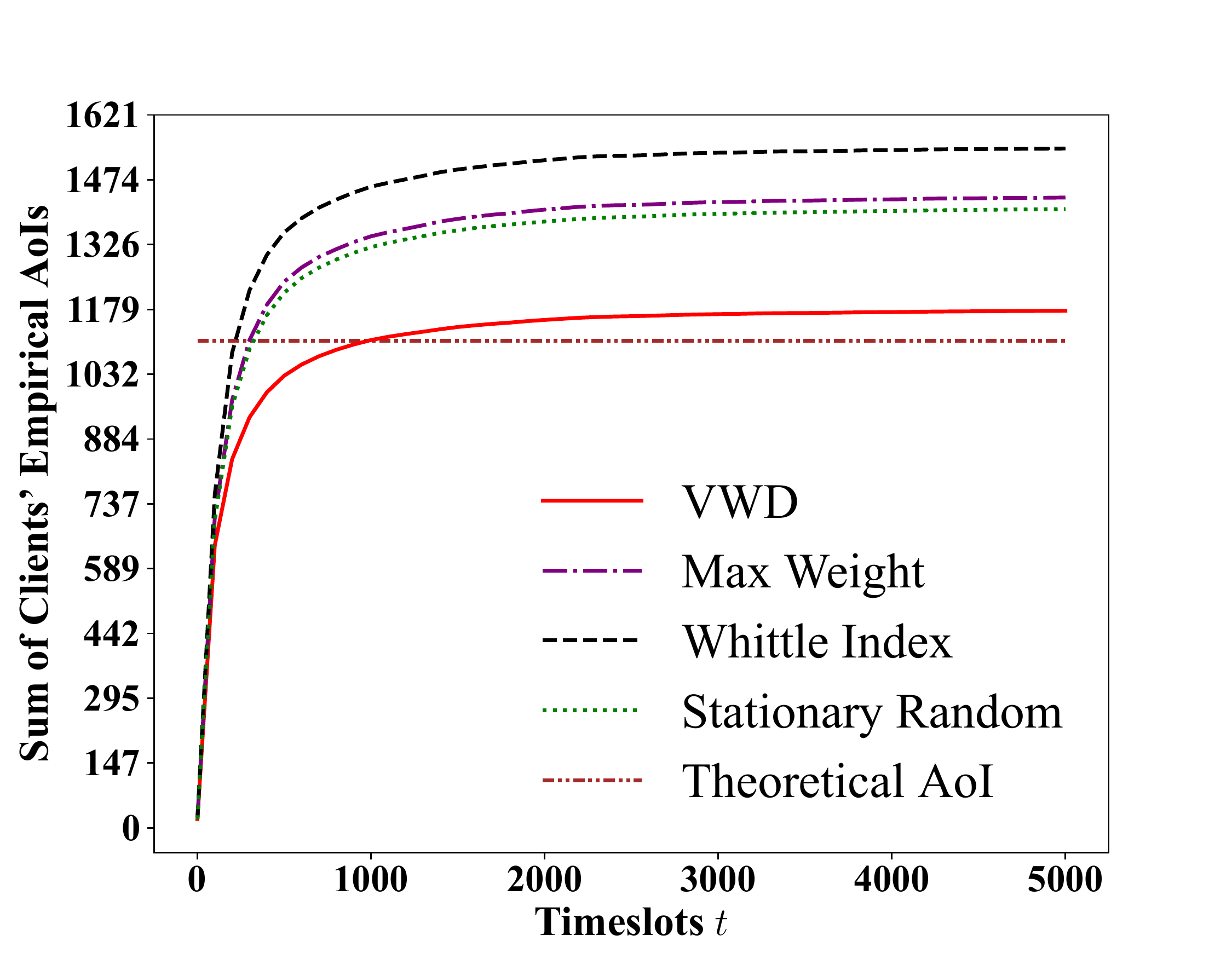}}
\end{center}
\caption{Total Uniformly Weighed Empirical Age of Information (AoI) Averaged Over $1000$ Runs.}
\label{fig:empirical_aoi_sum}
\end{figure*}

We conclude this section by discussing how to leverage Theorems~\ref{theorem:inner bound} and \ref{thm: approach solution} to solve the second-order network optimization problem. Recall that the performance of a client $i$ is modeled by $F_i(\mu_i, \sigma_i^2)$. For example, when the goal is to minimize total AoI, we can define $F_i(\mu_i, \sigma_i^2) = -\frac{1}{2}(\frac{\sigma_i^2}{\mu_i^2}+\frac{1}{\mu_i})-\frac{1}{\lambda_i}+\frac{1}{2}$. Hence, the second-order optimization problem can be written as the following:

\begin{align} \label{eq:objective_function}
    \max & \sum_{i=1}^NF_i(\mu_i,\sigma_i^2)\\
    \mbox{s.t. } & \mbox{(\ref{eq:sufficient:mean}) -- (\ref{eq:sufficient:non-negative})}.
\end{align}

The condition (\ref{eq:sufficient:mean}) involves strict inequalities, which cannot be used by standard optimization solvers. We change (\ref{eq:sufficient:mean}) to $\sum_{i\in S}\mu_i\leq m_S-\delta$, where $\delta$ is a small positive number. After the change, the optimization problem can be directly solved by standard solvers to find the optimal $\{\mu_i, \sigma_i^2|1\leq i\leq N\}$. After finding the optimal $\{\mu_i, \sigma_i^2|1\leq i\leq N\}$, one can use the VWD policy to attain the optimal network performance.

\section{Simulation Results} \label{sec:simulation}

In this section, we present the simulation results for the proposed scheduler VWD. The objective is to minimize the total weighted AoI, $\sum_i \alpha_i \overline{AoI}_i$, where $\alpha_i$ is the weight of client $i$. The system model is the one discussed in Section~\ref{sec: aoi}. Each client has a Gilbert-Elliott channel with transition probabilities $p_i$ and $q_i$. In each time slot, each client $i$ generates a new packet with probability $\lambda_i$. VWD is evaluated against three recent scheduling policies on this problem. We provide a description of each policy, along with modifications needed to fit the testing setting.

\begin{itemize}
    \item \textbf{Whittle index policy}: This policy is based on the Whittle index policy in  \cite{Hsu18}. Under our setting, the policy calculates an index for ON clients based on their AoIs as $W_i(t) = \frac{AoI_i^2(t)}{2} - \frac{AoI_i(t)}{2} + \frac{AoI_i(t)}{q_i/(p_i+q_i)}$, and then schedules the ON client with the largest index. \cite{Hsu18} has shown that $W_i(t)$ is indeed the Whittle index of a client when the channel is i.i.d., i.e., $p_i+q_i=1$, and $\lambda_i=1$.
    
    \item \textbf{Stationary randomized policy}: This policy calculates a weight $\mu_i$ for each client. In each time slot, it randomly picks an ON client, with the probability of picking $i$ being proportional to $\mu_i$. In the setting of \cite{kadota2019minimizing}, it has been shown that, when $\mu_i$ is properly chosen, this policy achieves an approximation ratio of four in terms of total weighted AoI. In our setting, we choose $\mu_i$ to be the optimal $\mu_i$ from solving (\ref{eq:objective_function}). 
    
    \item \textbf{Max weight policy} \cite{kadota2019minimizing}: This policy schedules the ON client with the largest $(AoI_i(t) - z_i(t))/\mu_i$. In the setting of \cite{kadota2019minimizing}, $z_i(t)$ is the time since client $i$ generates the latest packet. It has been shown that the total weighted AoI under this policy is no larger than that under the stationary randomized policy, and therefore this policy also achieves an approximation ratio of four. In our setting, the AP does not know when each client generates a new packet. Hence, we choose $z_i(t)$ to be $\frac{1}{\lambda_i}$, which is the expected time since client $i$ generates the latest packet.
\end{itemize}

\begin{figure*}[hbt!]
\begin{center} 
\subfigure[N = 5 Clients.]{\includegraphics[width=0.3\textwidth, height=1.8in]{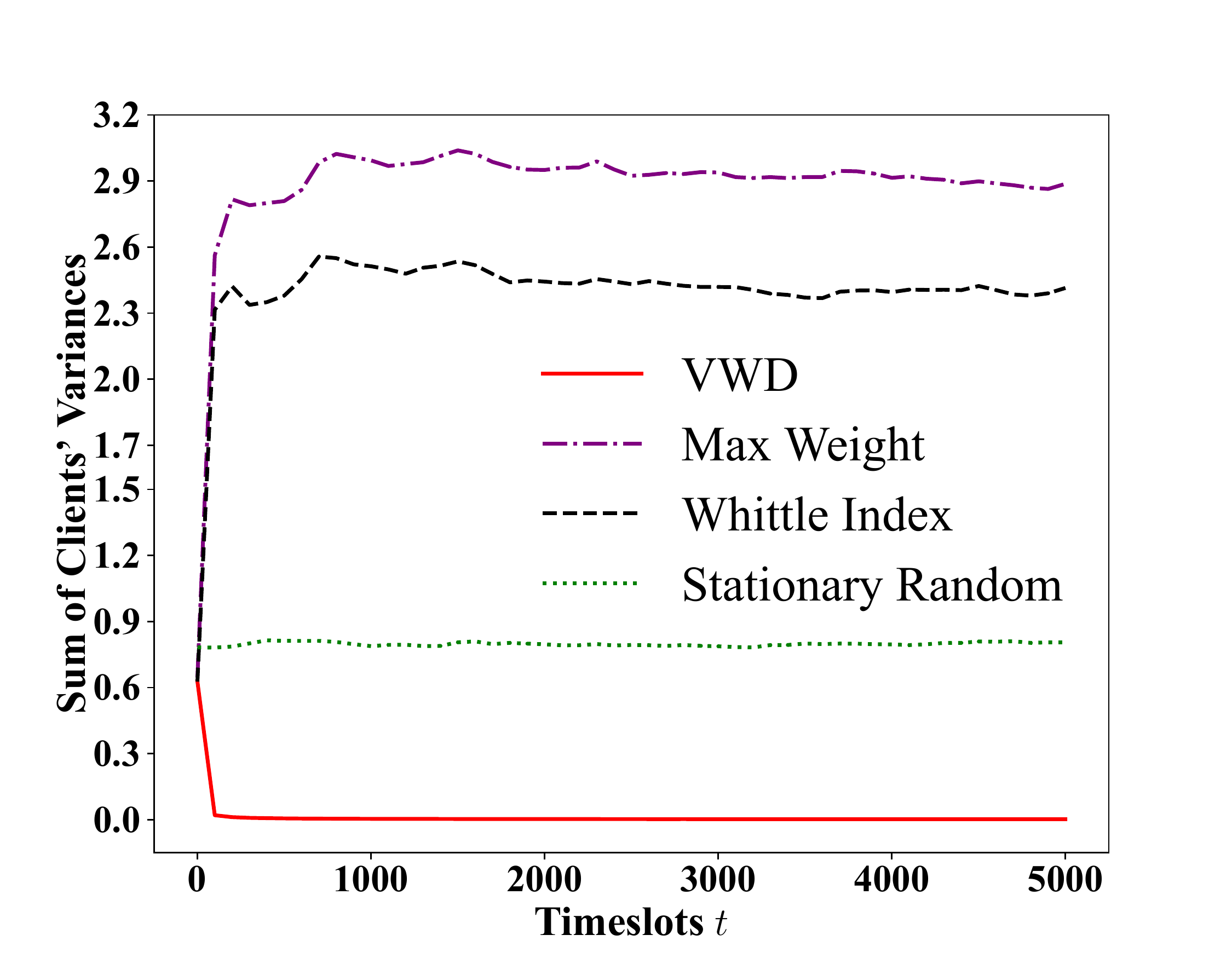}}
\subfigure[N = 10 Clients.]{\includegraphics[width=0.3\textwidth, height=1.8in]{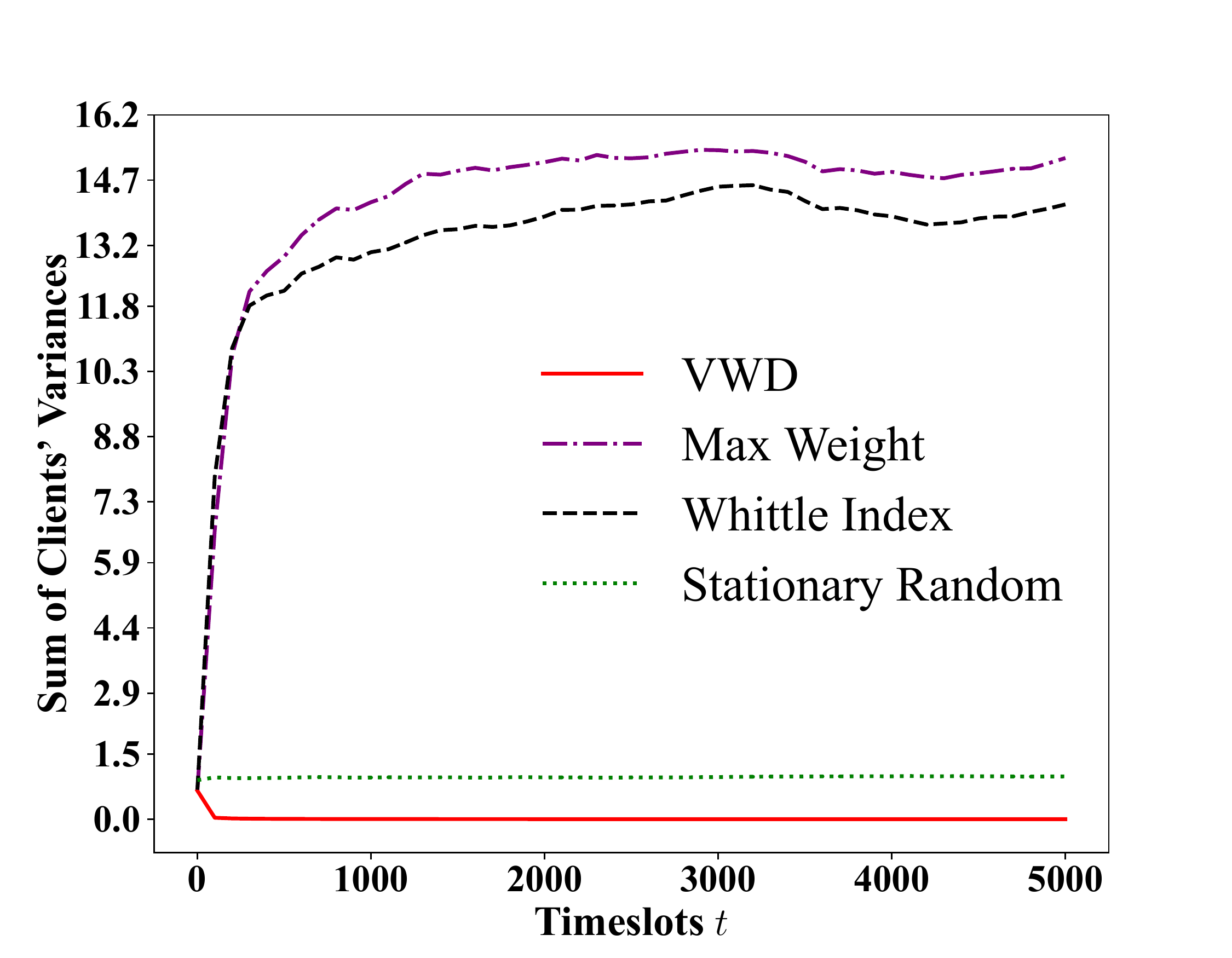}}
\subfigure[N = 20 Clients.]{\includegraphics[width=0.3\textwidth, height=1.8in]{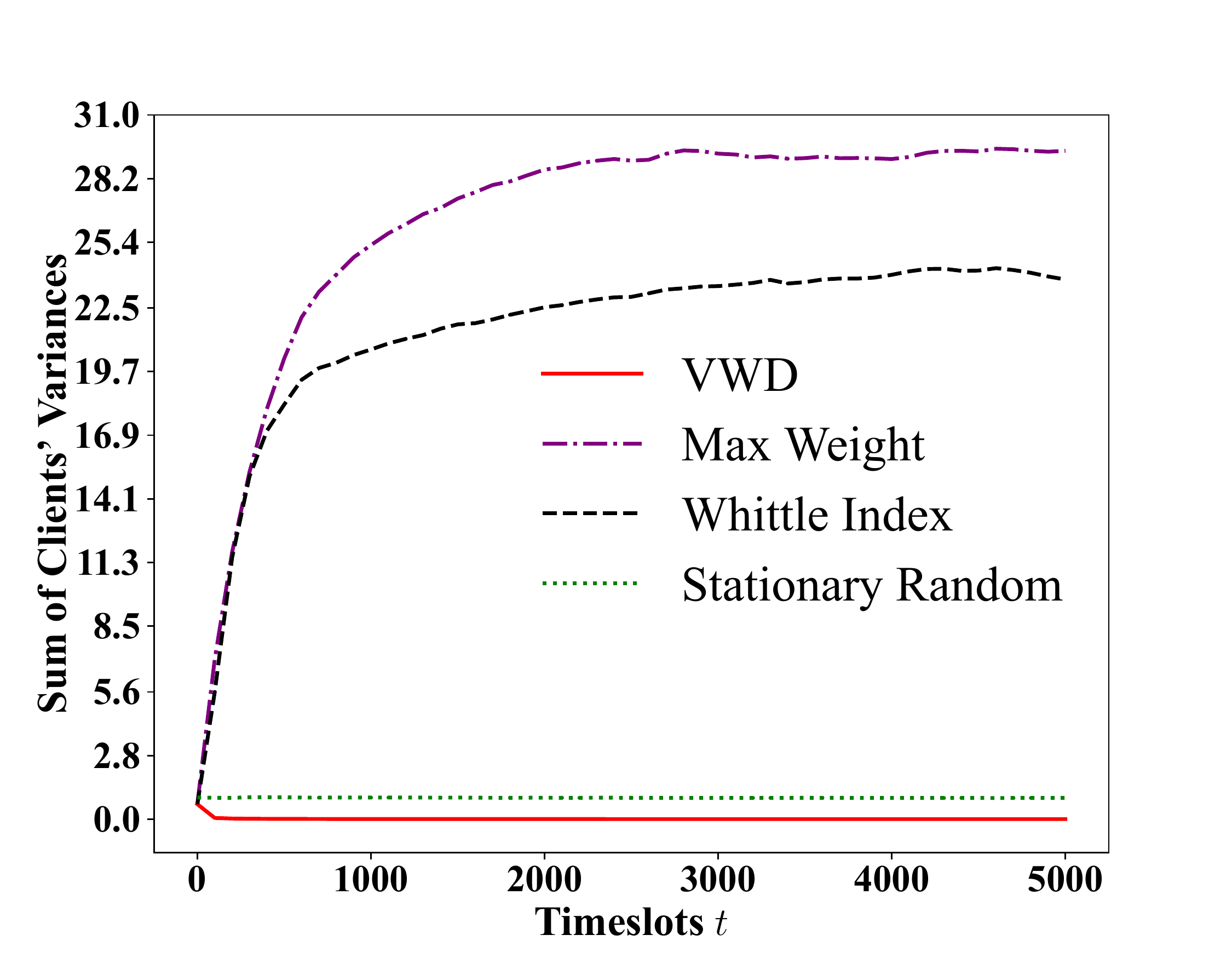}}
\end{center}
\caption{Empirical Variance of All Clients Averaged Over $1000$ Runs.}
\label{fig:variance_policies}
\end{figure*}

\begin{figure*}[hbt!]
\begin{center} 
\subfigure[N = 5 Clients.]{\includegraphics[width=0.3\textwidth, height=1.8in]{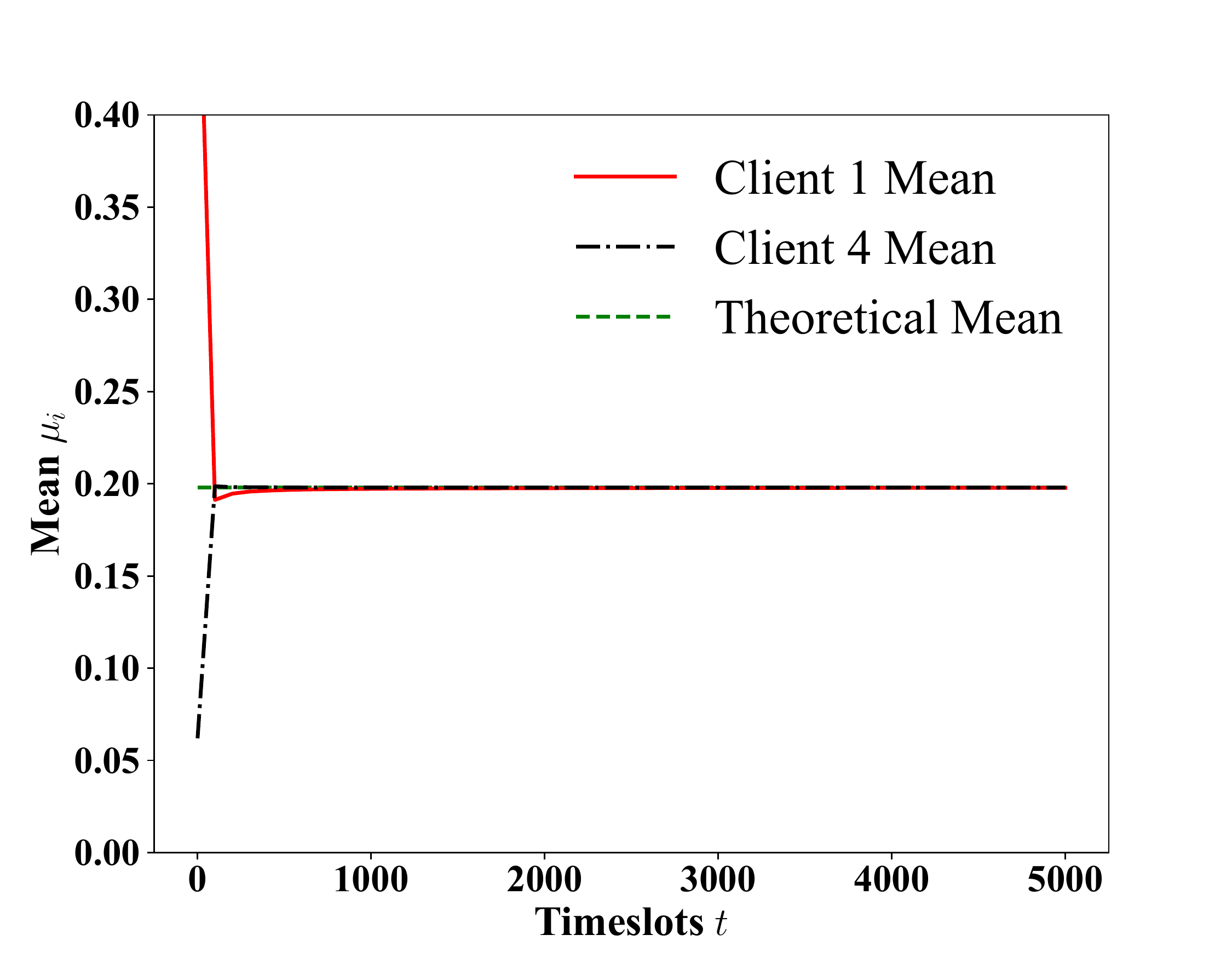}}
\subfigure[N = 10 Clients.]{\includegraphics[width=0.3\textwidth, height=1.8in]{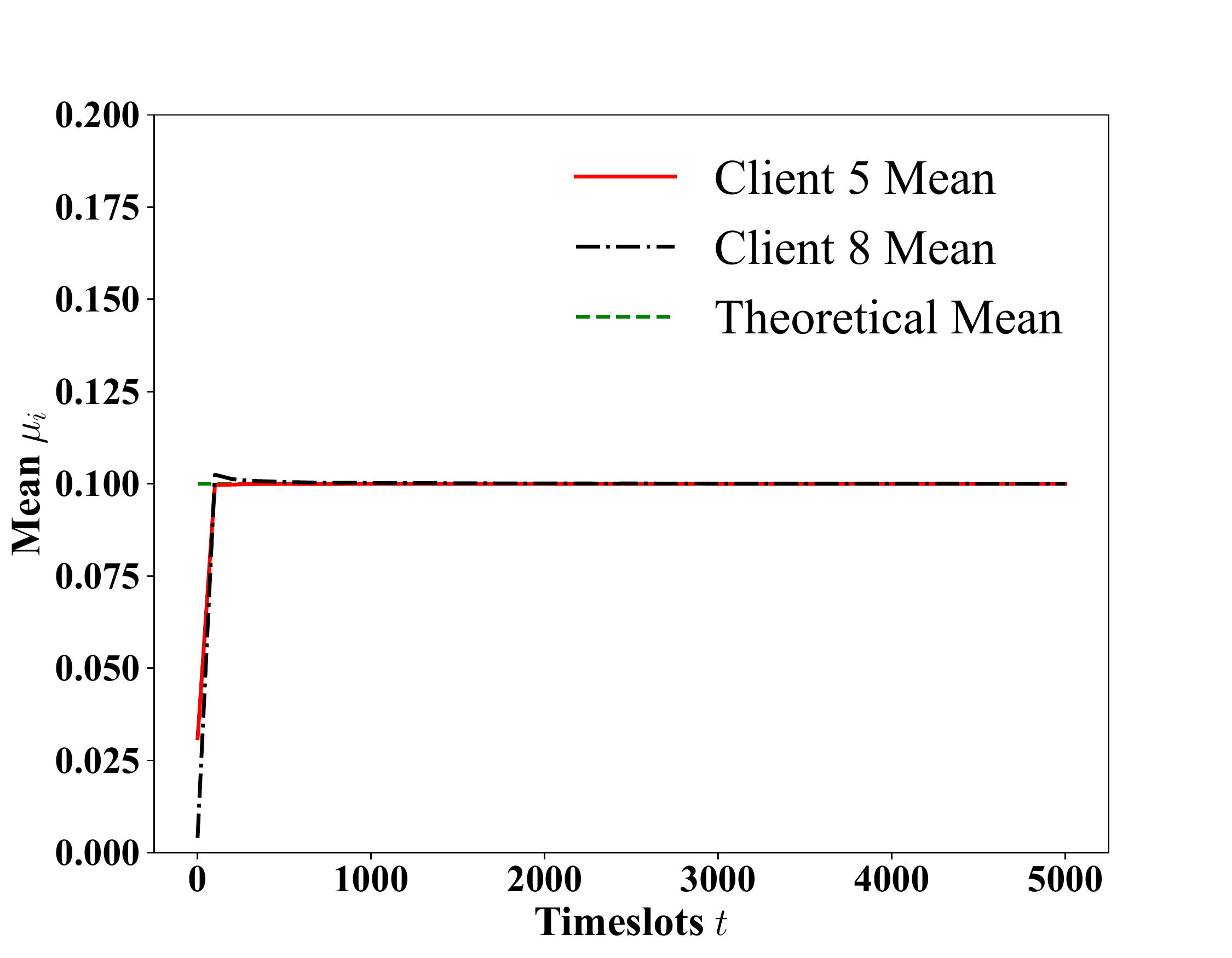}}
\subfigure[N = 20 Clients.]{\includegraphics[width=0.3\textwidth, height=1.8in]{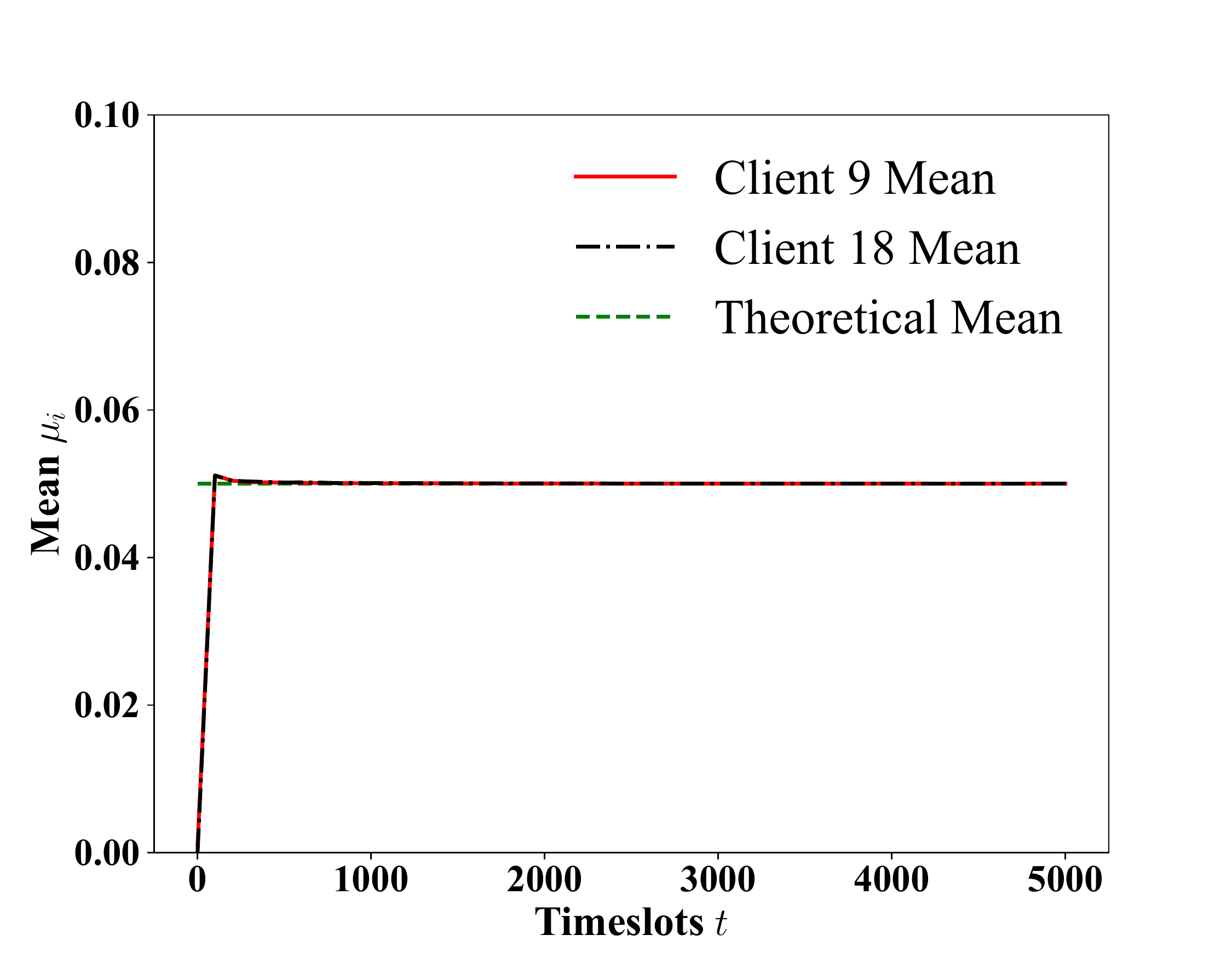}}
\end{center}
\caption{Mean Convergence of Two Randomly Selected Clients.}
\label{fig:mean_convergence}
\end{figure*}

\begin{figure*}[hbt!]
\begin{center} 
\subfigure[N = 5 Clients.]{\includegraphics[width=0.3\textwidth, height=1.8in]{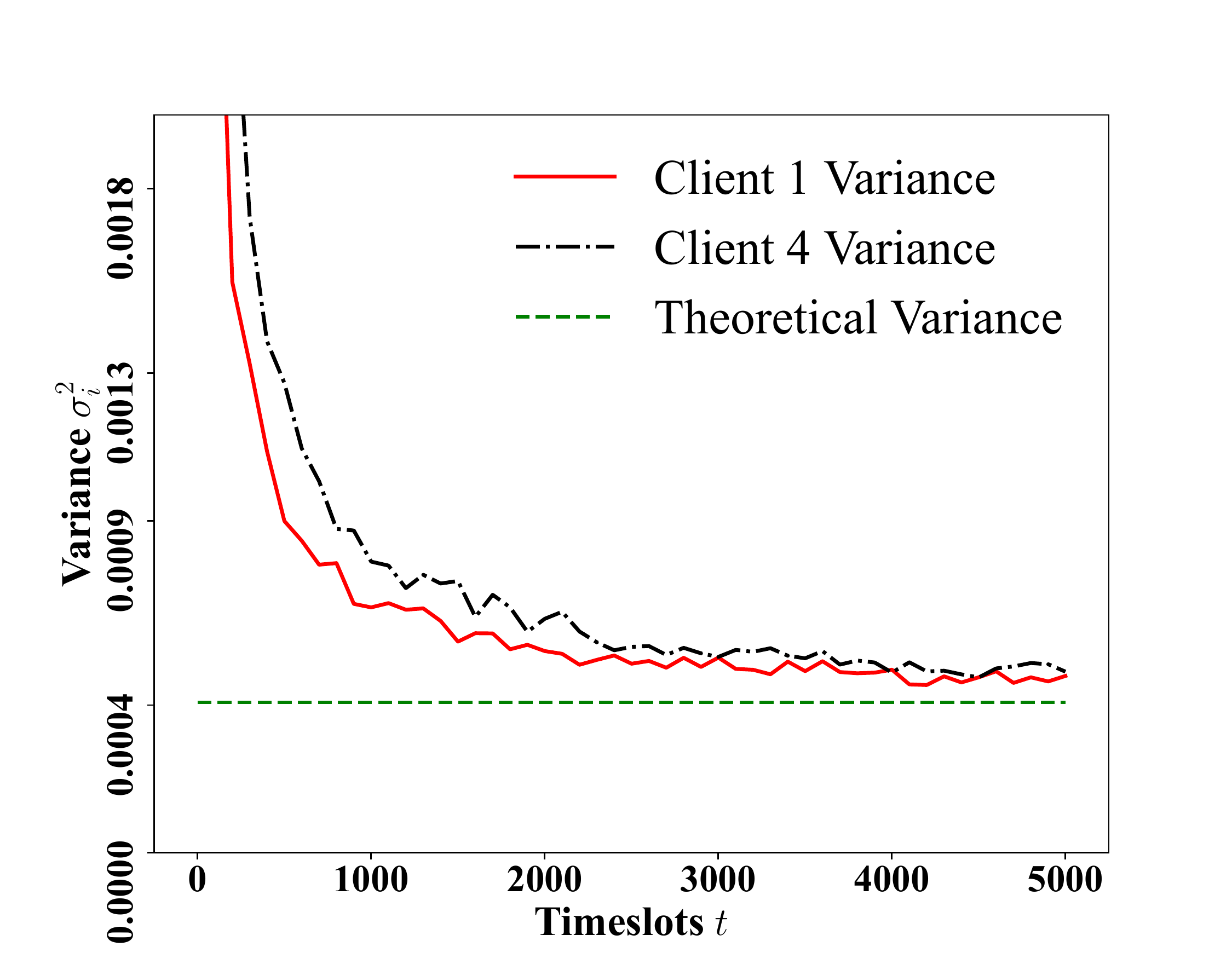}}
\subfigure[N = 10 Clients.]{\includegraphics[width=0.3\textwidth, height=1.8in]{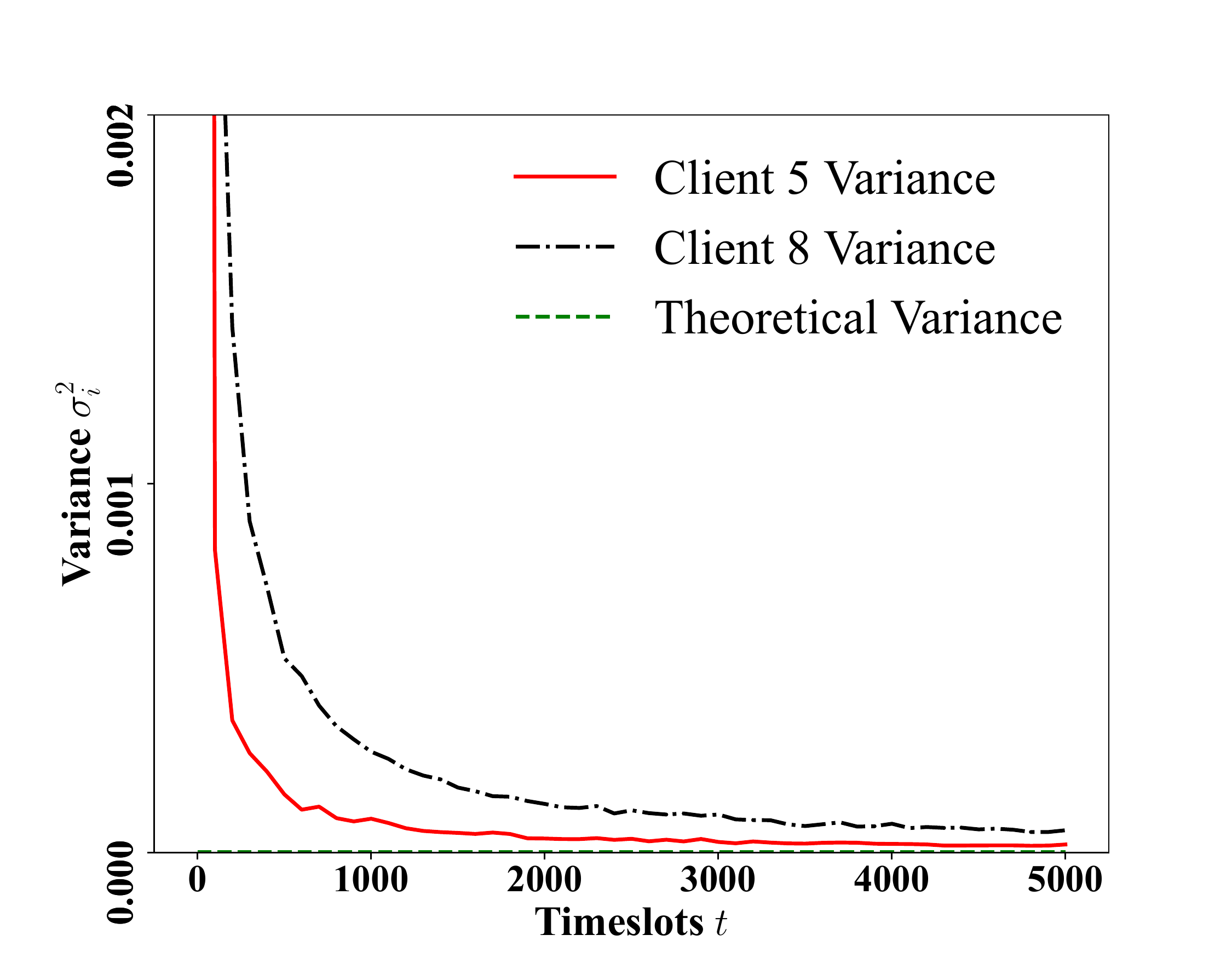}}
\subfigure[N = 20 Clients.]{\includegraphics[width=0.3\textwidth, height=1.8in]{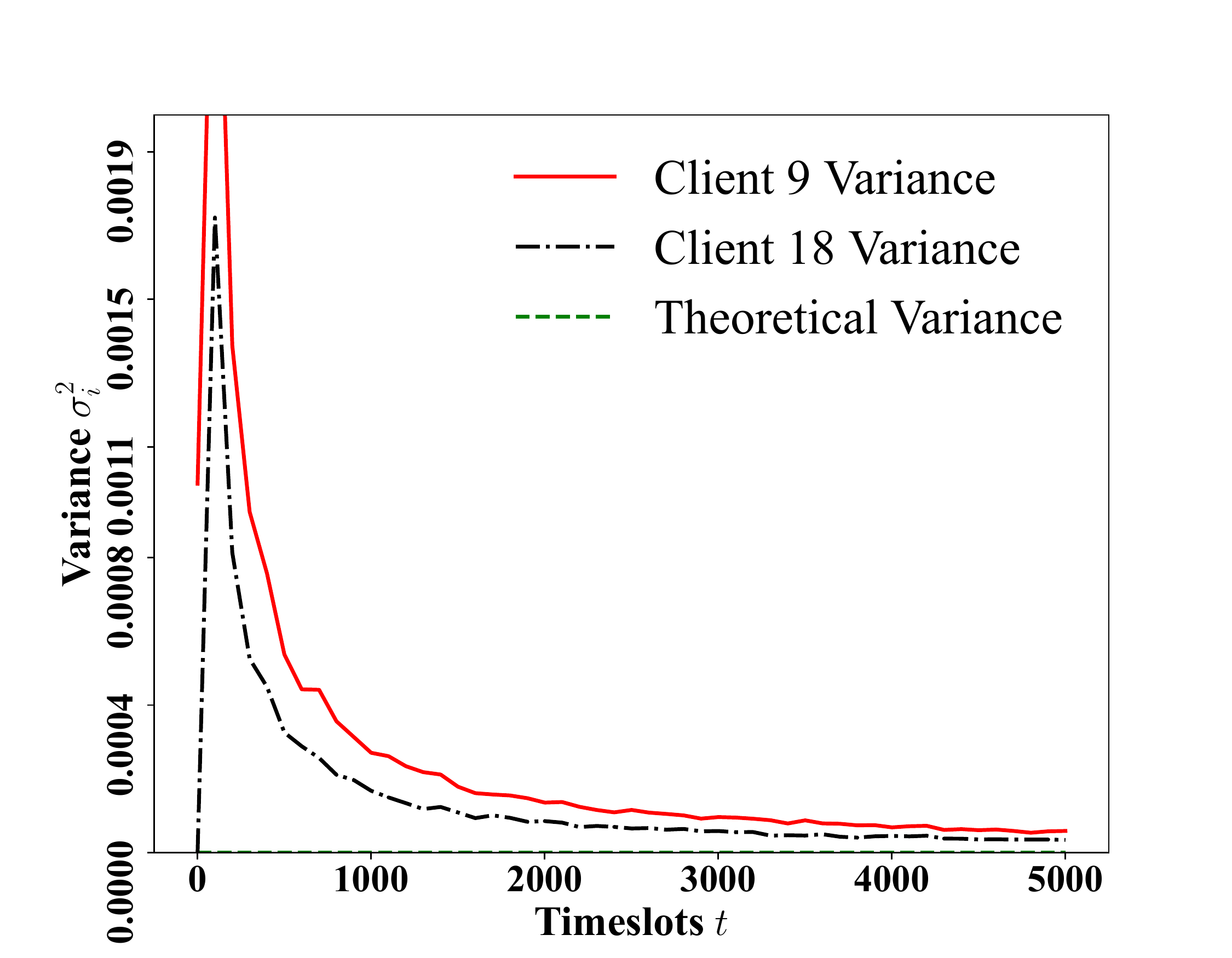}}
\end{center}
\caption{Variance Convergence of Two Randomly Selected Clients.}
\label{fig:var_convergence}
\end{figure*}

We consider three different systems, each with 5 clients, 10 clients, and 20 clients, respectively. For each system, $p_i$ and $q_i$ are randomly chosen from the range $(0.05, 0.95)$, and $\{\lambda_i\}$ is randomly chosen from $(\frac{0.1}{N}, \frac{1}{N})$. After determining the values of $p_i$, $q_i$ and $\lambda_i$, we generate $1000$ independent traces of channels and packet arrivals. The performance of each policy is the average over these 1000 independent traces. We consider both the unweighted case, i.e., $\alpha_i\equiv 1, \forall i$, and the weighted case. In addition to the evaluated policies, we also include the numerical solutions from solving the problem (\ref{eq:objective_function}), which is referred to as the Theoretical AoI.


\subsection{Empirical AoI Performance With Equal AoI Weights}

Fig.~\ref{fig:empirical_aoi_sum} shows the average total AoI for different network sizes $N = \{5, 10, 20\}$ when $\alpha_i\equiv 1$. It can be observed that VWD achieves the smallest total AoI in all systems, with max weight performing virtually the same as VWD when $N=5$. VWD's superiority becomes more significant as $N$ increases. It can also be observed that the empirical AoI under VWD is very close to the theoretical AoI based on the solution to (\ref{eq:objective_function}), and the difference decreases as $N$ increases. The differences between the empirical AoI under VWD and the theoretical one are $10.7\%$, $7.8\%$, and $6.1\%$ for $N=5, 10, 20$, respectively.

To understand why VWD performs much better than the other three policies, we evaluate the total empirical variance under each policy. Specifically, let $d_i(t)$ be the total number of packet deliveries for client $i$ from time 1 to time $t$. The empirical variance of a client $i$ at time $t$ is defined as the variance of $\frac{d_i(t)}{\sqrt{t}}$ across all 1000 independent runs. The total empirical variance is then the sum of the empirical variances of all clients. Fig.~\ref{fig:variance_policies} shows that VWD has much smaller variances than the other three policies. The ability to properly control variance enables VWD to achieve small AoIs.

\begin{figure*}[hbt!]
\begin{center} 
\subfigure[N = 5 Clients.]{\includegraphics[width=0.3\textwidth, height=1.8in]{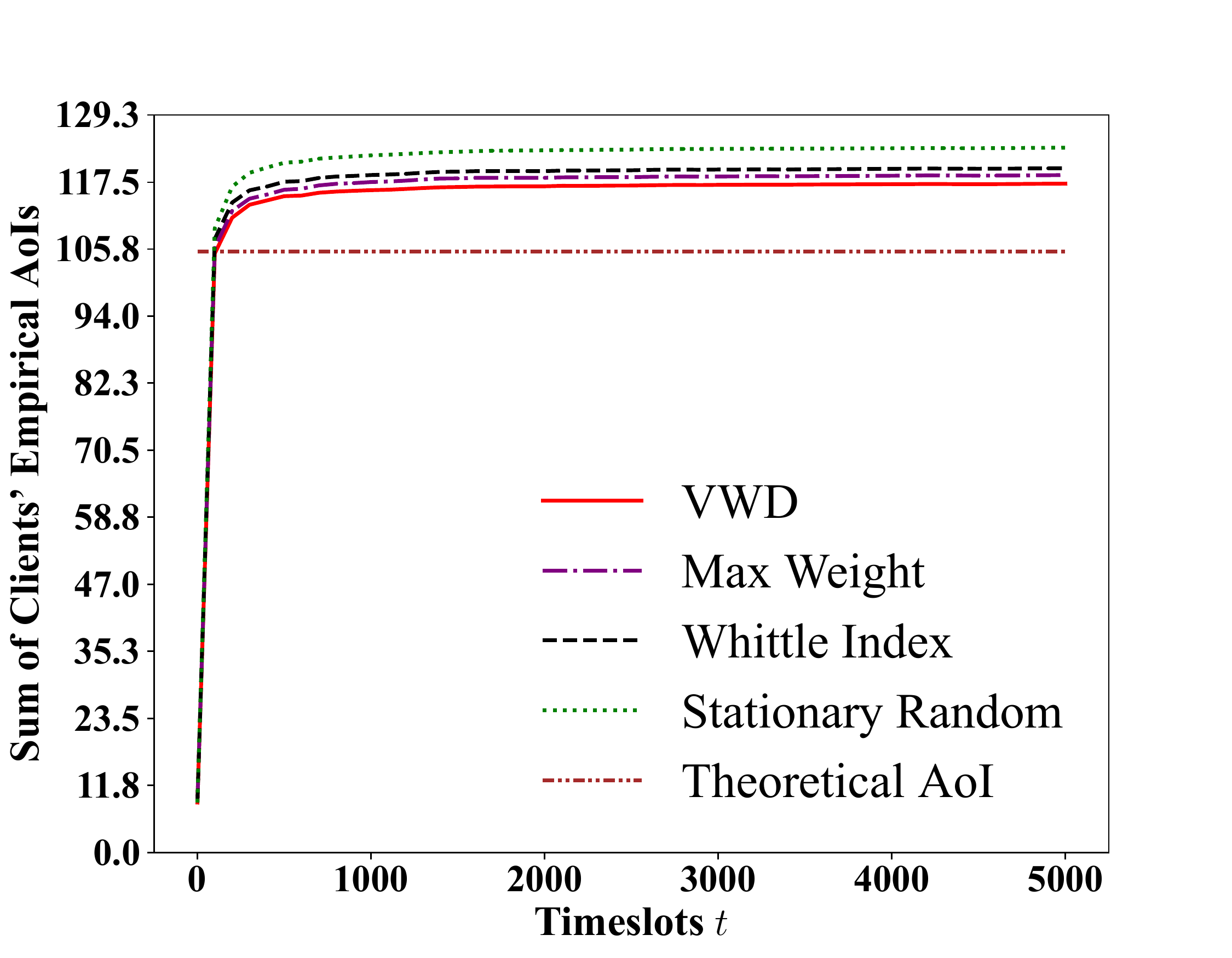}}
\subfigure[N = 10 Clients.]{\includegraphics[width=0.3\textwidth, height=1.8in]{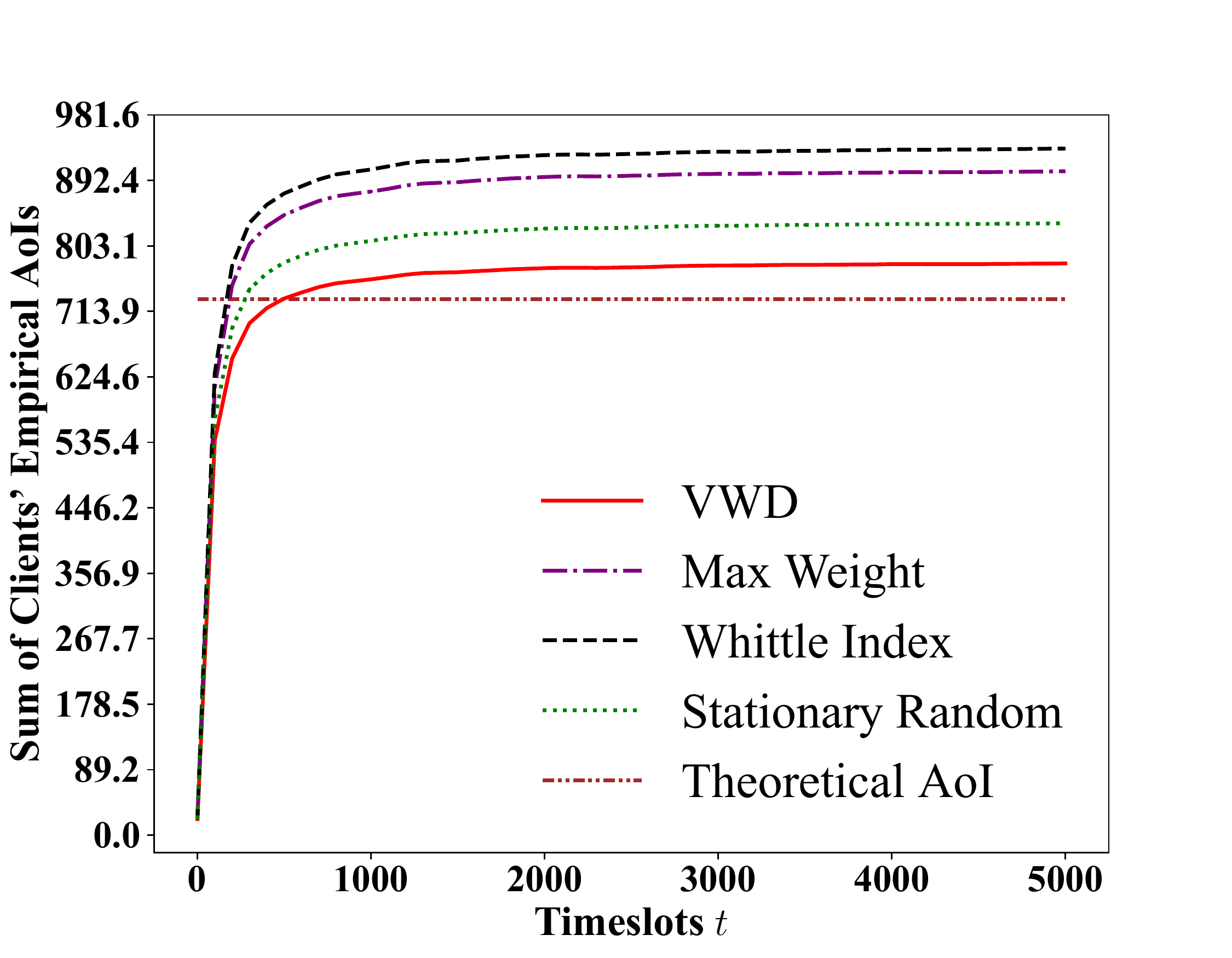}}
\subfigure[N = 20 Clients.]{\includegraphics[width=0.3\textwidth, height=1.8in]{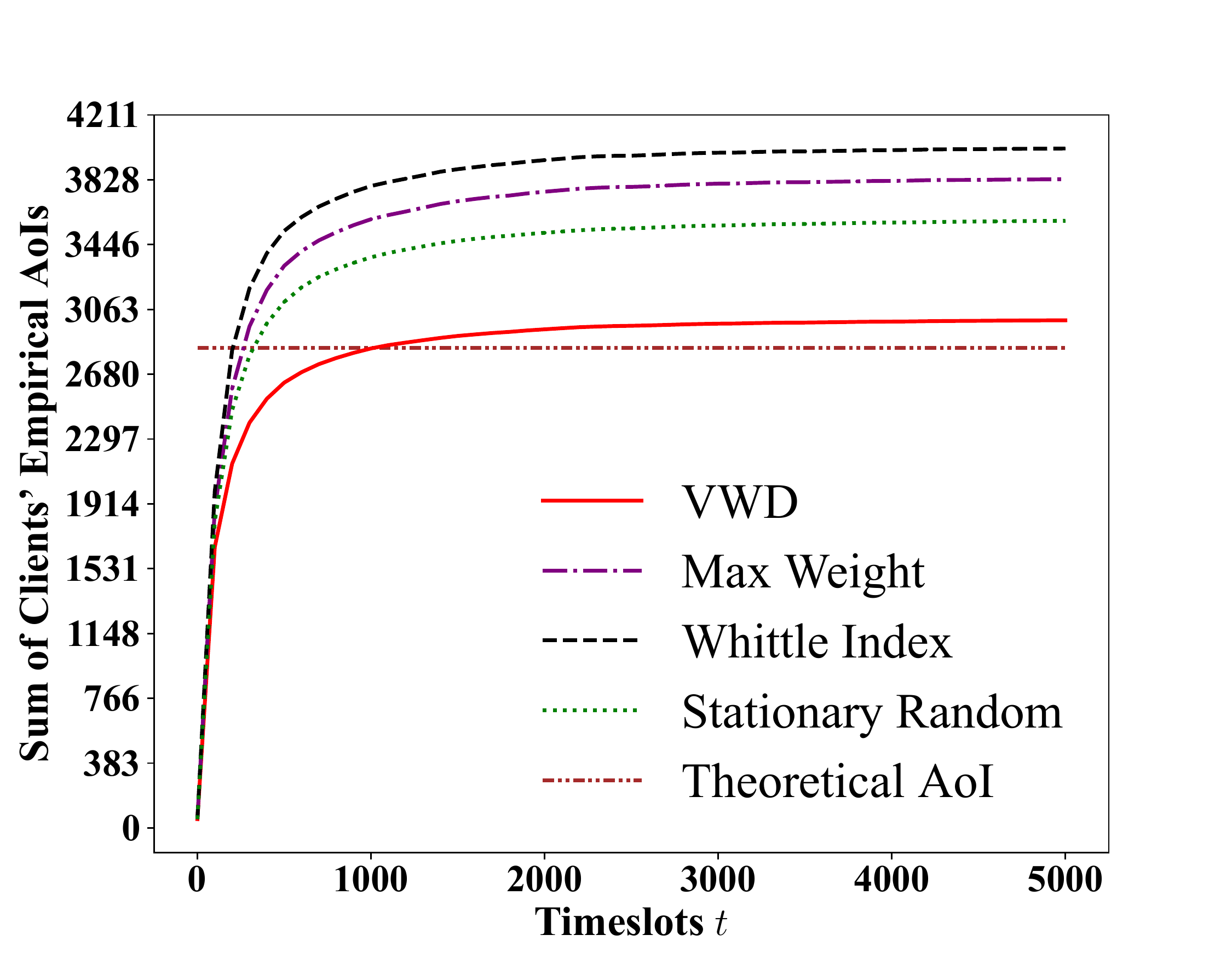}}
\end{center}
\caption{Total Weighted Empirical Age of Information (AoI) Averaged Over $1000$ Runs.}
\label{fig:weighted_aoi}
\end{figure*}

We also evaluate the convergence time of VWD. For each system, we randomly select two clients and plot their empirical means, i.e., the average of $\frac{d_i(t)}{t}$ across all independent runs, and empirical variances. Since the objective is to minimize the unweighted sum of AoIs, the optimal solution to (\ref{eq:objective_function}) has $\mu_i=\mu_j$ and $\sigma_i^2=\sigma_j^2$ for all $i\neq j$. We call the optimal $\mu_i$ and $\sigma_i^2$ obtained from solving (\ref{eq:objective_function}) the theoretical mean and the theoretical variance, respectively. The results are shown in Figs.~\ref{fig:mean_convergence} and \ref{fig:var_convergence}. It can be observed that both the empirical means and the empirical variances of clients indeed converge to their respective theoretical values. The empirical means converges to the theoretical ones very fast. On the other hand, it takes up to $355$ slots for the empirical variances to be within 0.001 from the theoretical variances. This convergence time may be the reason why the empirical AoI is larger than the theoretical one.

\subsection{Weighed Total AoI Evaluation}

We now present the results for the weighted AoI. The weights $\alpha_1, \alpha_2,\dots$ are randomly chosen from the range $(1,5)$ and independently from each other. All other parameters are the same as in the unweighted case. Fig.~\ref{fig:weighted_aoi} shows results for network sizes $N = \{5, 10, 20\}$.
VWD still outperforms other policies for all tested systems. Similar to the unweighted case, it can be observed that the superiority of VWD becomes more significant, and the gap between VWD and theoretical AoI becomes smaller, with more clients in the system. 

\section{Related Works} \label{sec:related}
There have been many works on scheduling in wireless networks for minimizing AoI. In~\cite{tripathi17}, the Tripathi and Moharir schedule over multiple orthogonal channels and propose Max-Age Matching and Iterative Max-Age Scheduling, which they show to be asymptotically optimal.
Hsu, Modiano and Duan~\cite{Hsu17} studied the problem of scheduling updates for multiple clients where the updates arrive i.i.d. Bernoulli, and formulate the Markov decision process (MDP) and prove structural results and finite-state approximations. In~\cite{Hsu18}, Hsu follows up this work by showing that a Whittle index policy can achieve near optimal performance with much lower complexity.
Sun et al.~\cite{Sun18} studied scheduling for multiple flows over multiple servers, and show that maximum age first (MAF)-type policies are nearly optimal for i.i.d. servers.
In~\cite{Talak20}, Talak, Karaman and Modiano study scheduling a set of links in a wireless network under general interference constraints.
The optimization of AoI and timely-throughput were studied in~\cite{Kadota18a,Lu18}. All of these works assume i.i.d channels.

There have been a limited number of works on Markov channel and source models related to AoI. In the recent work~\cite{pan20}, Pan et al. study scheduling a single source and choosing between a Gilbert-Elliott channel and a deterministic lower rate channel.
Buyukates and Ulukus~\cite{buyukates20} study the age-optimal policy for a system where the server is a Gilbert-Elliott model and one where the sampler follows a Gilbert-Elliott model.
In~\cite{nguyen19}, Nguyen et al. analyze the Peak Age of Information (PAoI) of a two-state Markov channel with differing cases of channel state information (CSI) knowledge.
Kam et al.~\cite{kam18} study the remote estimation of a Markov source, and they propose effective age metrics that capture the estimation error. Our work differs in that we focus on scheduling for multiple clients from a single AP over parallel non-i.i.d. channels. 

There have been some recent efforts on  studying short-term performance through Brownian motion approximation~\cite{hsieh16,hou17,hsieh20,guo21}, but each of them is limited to a specific channel model and a specific application. 

\section{Conclusion} \label{sec:conclusion}

In this paper, we presented a theoretical second-order framework for wireless network optimization. This framework captures the behaviors of all random processes by their second-order models, namely, their means and temporal variances. We analytically established a simple expression of the second-order capacity region of wireless networks. A new scheduling policy, VWD, was proposed and proved to achieve every interior point of the second-order capacity region. The framework utility is demonstrated by applying it to the problem of AoI optimization over Gilbert-Elliott channels. We derived closed-form expressions of second-order models for both Gilbert-Elliott channels and AoIs, and formulated the problem of minimizing weighted total AoI as an optimization problem over the means and temporal variances of delivery processes. The solution of this optimization problem can then be used as parameters for VWD. Simulation results show that VWD achieves much smaller weighted total AoI than other policies.

\newpage
\bibliographystyle{IEEEtran}
\bibliography{reference}
\end{document}